\documentclass[format=acmsmall, review=false, screen=true]{acmart}

\usepackage{booktabs} 

\usepackage[ruled]{algorithm2e} 

\SetAlFnt{\small}
\SetAlCapFnt{\small}
\SetAlCapNameFnt{\small}
\SetAlCapHSkip{0pt}
\IncMargin{-\parindent}

\usepackage{amsmath}
\usepackage{amssymb}
\usepackage{bm}

\newcommand*{\Cdot}{\raisebox{-0.45ex}{\scalebox{1.15}{$\cdot$}}}
\DeclareMathOperator*{\argmax}{arg\,max}

\graphicspath{{./images/}}

\acmJournal{TKDD}
\acmVolume{9}
\acmNumber{4}
\acmArticle{39}
\acmYear{2017}
\acmMonth{3}
\copyrightyear{2017}

\setcopyright{acmlicensed}

\acmDOI{0000001.0000001}

\received{March 2017}

\begin{document}
\title{Spatio-Temporal Modeling of Users' Check-ins in Location-Based Social Networks}

\author{Ali Zarezade}
\author{Sina Jafarzadeh}
\author{Hamid R. Rabiee}
\affiliation{%
  \institution{Sharif University of Technology }
  \streetaddress{Azadi Ave.}
  \city{Tehran}
  \country{Iran}
  }

\begin{abstract}
Social networks are getting closer to our real physical world. 
People share the exact location and time of their check-ins and are influenced by their friends. Modeling the spatio-temporal behavior of users in social networks is of great importance for predicting the future behavior of users, controlling the users' movements, and finding the latent influence network.
It is observed that users have periodic patterns in their movements. Also, they are influenced by the locations that their close friends recently visited.
Leveraging these two observations, we propose a probabilistic model based on a doubly stochastic point process with a periodic decaying kernel for the time of check-ins and a time-varying multinomial distribution for the location of check-ins of users in the location-based social networks.
We learn the model parameters using an efficient EM algorithm,  which distributes over the users. 
Experiments on synthetic and real data gathered from Foursquare show that the proposed inference algorithm learns the parameters efficiently and our model outperforms the other alternatives in the prediction of time and location of check-ins.
\end{abstract}

%
%

%
%



\thanks{
This work is supported by ICT Innovation Center, Department of Computer Engineering, Sharif University of Technology, Tehran, Iran
}

\maketitle

\renewcommand{\shortauthors}{G. Zhou et al.}

\section{Introduction}
The advances in location-acquisition techniques and the proliferation of mobile devices have generated an enormous amount of spatial and temporal data of users activities \cite{zheng2011location}.
People can upload a geotagged video, photo or text to social networks like Facebook and Twitter, share their present location on Foursquare or share their travel route using GPS trajectories to GeoLife \cite{zheng2010geolife}.
A considerable amount of this spatio-temporal data is generated by the activity of users in location-based social networks (LBSN). 
In a typical LBSN, like Foursquare, users share the time and geolocation of their check-ins, comment about it, or unlock badges by exploring new venues. 

Many techniques have been proposed for processing, managing, and mining the trajectory data in the past decade \cite{zheng2015trajectory}. Several other studies try to leverage the spatial data in recommender systems \cite{bao2014recommendations}.  However, a few works have attempted to model the spatio-temporal behavior of users in LBSNs \cite{cho2013latent,cho2011friendship}. Given the history of users' check-ins, the goal is to predict the time and location of each user's check-in utilizing a model.
This model can also be used to find the influence network between users which made up of their check-ins, detect the influential users and popular locations, predict the peak hours of a restaurant, recommend a location, and even control the movement of users. 

In this paper, we propose a probabilistic spatio-temporal generative model for the check-ins of users in location-based social networks, which can be used in predicting the future check-ins of the users, and discovering the latent influence network. 
People usually have periodic patterns in their movements  \cite{williams2013periodic,cho2011friendship,li2010mining}. For example, a typical user may check into her office in the morning and to a nearby restaurant at noon then return home and repeat this behavior in the following days. 
We model the time of check-ins of each user with a novel periodic decaying doubly stochastic point process which leverages the periodicity in the movements of users and can also capture any drift in their patterns.
To model the location of check-ins we use the fact that users in social media are influenced by the activities of their friends \cite{zarezade2015correlated,Farajtabar2014,gomez2015estimating}. 
If many of your close friends have checked into a specific restaurant recently, then there is a high probability that you select that restaurant, next time.
We model the location of check-ins using a time-varying multinomial distribution.
In summary, we propose:
\begin{itemize}
    \item a periodic point process for modeling the time of users' check-ins, which captures the periodic behavior in the movement of users,
    \item a time-varying multinomial distribution for modeling the location of users' check-ins, which incorporates the mutually-exciting effect of the friends' history of check-ins,
    \item a scalable inference algorithm based on the EM algorithm to find the model parameters, which is distributed over users,
    \item a compelling dataset of Foursquare users' check-ins, curated from 12000 active users during three months in the year 2015.
\end{itemize}

\section{Prior Works}
Modeling information diffusion in social networks has attracted a lot of attentions in recent years \cite{Rodriguez2011,Farajtabar2014,gomez2016influence,du2016recurrent}. 
Given the times that users have adopted to a contagion (information, behavior, or meme), the problem is to model the time and the user of next adoption, \textit{i.e.}, predict the next event.
Early methods \cite{Rodriguez2011,du2012learning,Rodriguez2013} studied information diffusion using a pair-wise probability distribution for each link from node $j$ to $i$, which is the probability that node $i$ generates an event in time $t_i$ due to the event of node $j$ at time $t_j$. These methods overlook the external effects on the generation of events. In addition, they assume that each node adopts a contagion at most once, \textit{i.e.}, events are not recurrent. 
These issues were later addressed in 
\cite{iwata2013discovering,cho2013latent,yang2013mixture,yuan2013and,linderman2014discovering,valera2015,tran2015netcodec,he2015hawkestopic,gomez2016influence}, 
which use point processes for the modeling of events.
In \cite{yang2013mixture,Rodriguez2013,iwata2013discovering,linderman2014discovering}, cascades are assumed to be independent and are modeled by a special point process, called Hawkes \cite{Hawkes1971}. The independence assumption is removed in \cite{valera2015,zarezade2015correlated}, they tried to model the correlation between multiple competing or cooperating cascades.
Other studies  \cite{yang2013mixture,he2015hawkestopic,tran2015netcodec,gui2014modeling,hosseini2016hnp3}, use the additional information of the diffusion network such as topic of tweets or the community structure to better model the influence network.
Most of the previous works studied the information diffusion on  microblogging networks like Twitter, whereas we try to model the time and location of users' check-ins in the location-based networks like Foursquare. 

The prior works in location-based social networks can be categorized into three groups: location recommendation, trajectory mining and location prediction.
The main approaches in location recommendation systems \cite{bao2014recommendations} are: content-based which uses data from a user's profile and the features of locations \cite{ying2010mining,ye2011semantic,xiao2014inferring,postel2013point}; link-based, which applies link analysis models like PageRank to identify the experienced users and interesting locations \cite{zheng2009mining,cao2010mining,yoon2010smart,scellato2011exploiting,liu2013learning}; and collaborative filtering which infers  users' preferences from their historical behavior, like the location history \cite{xiao2014inferring,levandoski2012lars,yuan2013time,postel2013point,li2016potential}.
In trajectory data mining, the source of data is usually generated by the GPS.
 These works include; trajectory pattern mining to find the next location of an individual \cite{cho2011friendship,tang2012discovery,zheng2014online,lichman2014modeling}, anomaly detection to detect unexpected movement patterns  \cite{lee2008trajectory,liu2014fraud}, and trajectory classification to differentiate between trajectories of different states, such as motions, transportation modes, and human activities \cite{zheng2008learning}. A comprehensive review of these methods can be found in the recent survey \cite{zheng2015trajectory}.
We also discriminate our work from location recommendation and trajectory mining methods, because our goal is to model the check-ins of users not to recommend a location or to find the trajectory patterns of users with the position data of their routes.
In location prediction, the goal is to predict the next location, given the user's profile data and the history of check-ins  \cite{gomes2013will,yuan2013and,malmi2013foursquare,liu2016predicting}.
But these methods do not consider; the relation between friends (using the influence matrix), aging effect in the history of checkins (using decaying kernel), exogenous effects on users' decisions, and periodicity in users' movement patterns.

The most similar works to ours are: 
\cite{cho2013latent}, which propose a spatio-temporal model for the interactions between a pair of users, but we model the check-ins of each user not the pair-wise interactions; and \cite{liu2016predicting} which propose a discriminative method to predict the location of next check-in, but we propose a generative model for the location and time of check-ins.

\section{Preliminaries}
To model the time of occurrences of a phenomenon, which are called events, we can use point processes on the real line. The phenomena can be, an earthquake \cite{ogata1988statistical}, a viral disease \cite{Barabasi2015} or the spread of information over a network \cite{Rodriguez2013}. 
The sequence of events, as defined below, is the realization of a point process.
%
\begin{definition}[Point Process]
	\textit{
	Let $\{ t_i \}_{i\in\mathbb{N}}$ be a sequence of non-negative random variables such that $\forall i \in \mathbb{N}, \, t_i<t_{i+1}$, then we call $\{ t_i \}_{i\in\mathbb{N}}$ a point process on $\mathbb{R}$, and $\mathcal{F}_t = \{t_i \,|\, i\in\mathbb{N}, t_i<t \}$ as its history or filtration.
	}
\end{definition}
There are different equivalent descriptions for the point processes such as;  sequence of points $\{t_i\}$, sequence of intervals (duration process) $\delta t_i$, counting process $N(t)$, or intensity process $\lambda(t)$ \cite{Daley2002}. In the following, we briefly  explain each definition. 

The counting process $N(t)$ associated with the point process $\{t_i\}_{i \in \mathbb{N}}$, counts the number of events occurred before time $t$, \textit{i.e.},
$N(t) = \sum_{i \in \mathbb{N}} \mathbb{I}(t_i<t)$, 
where $\mathbb{I}(\cdot)$ is the indicator function\footnote{
The indicator function $\mathbb{I}(x \in A)$ is $1$ if $x \in A$, and is $0$  otherwise.}. The duration process $\delta t_i$ associated with the point process $\{t_i\}_{i \in \mathbb{N}}$ is defined as $\forall i \in \mathbb{N},\, \delta t_i = t_i -t_{i-1}$. Finally, the intensity process $\lambda(t)$ is defined as the expected number of events per units of time, which generally depends on the history:
\begin{align*}
	\lambda(t|\mathcal{F}_t) &= \lim_{dt \rightarrow 0} \frac{1}{dt} \mathbb{E} \left[ N(t,t+dt] \,\vert\, \mathcal{F}_t \right] \\
	&= \lim_{dt \rightarrow 0} \frac{1}{dt} \text{Pr}[N(t,t+dt] > 0 \,|\, \mathcal{F}_t]
\end{align*}
where $N(t,s] := N(s)-N(t)$. To evaluate the likelihood of a sequence of events, $f(t_1,t_2,\ldots,t_n)$, we can use the chain rule of probability, $f(t_1,t_2,\cdots,t_n)=\prod_i f(t_i|t_{1:i-1})$. Therefore, it suffice to describe only the conditionals, which are abbreviated to $f^*(t)$. According to the definition of point processes, we can write the probability of occurring the $(n+1)$'th event in time $t$ as:
\begin{align*}
	f^*(t) \, dt = \text{Pr}\left\{N(t_n,t]=0, N(t,t+dt]=1 \,|\, t_{1:n}\right\}.
\end{align*}
If we divide both sides of the above equation by $1-F^*(t)$, where $F^*(\cdot)$ is the cdf of $f^*(\cdot)$, then in the limit as $dt \rightarrow 0$, we have:
\begin{align}
\frac{f^*(t) \, dt}{1-F^*(t)} &= \frac{\text{Pr}\left\{N(t_n,t]=0, N(t,t+dt]=1 \,|\, t_{1:n} \right\}} {\text{Pr}\left\{N(t_n,t]=0 \,|\, t_{1:n} \right\}} \nonumber \\
&=\text{Pr}\left\{N(t,t+dt]=1 \,|\, t_{1:n}, N(t_n,t]=0\right\} \nonumber \\
&=\text{Pr}\left\{N(t,t+dt]>0 \,|\, \mathcal{F}_{t}\right\} \nonumber
\end{align}
Therefore,  according to the definition of intensity, we find the relation between conditional distribution of the time of events and the intensity function as:
\begin{align}
	 \lambda^*(t) = \frac{f^*(t)}{1-F^*(t)}
\end{align}
where we use ${}^*$ superscript to show that a function is dependent on the history.
We can also express the relation of $\lambda^*(t)$ and $f^*(t)$ in the reverse direction \cite{Aalen2008}:
\begin{align}
f^*(t) &= \lambda^*(t) \exp\left(-\int_{t_n}^{t} \lambda^*(s) ds\right)
\end{align}
Now, the cdf can be easily evaluated:
\begin{align}
F^*(t) &= 1- \exp\left(-\int_{t_n}^{t} \lambda^*(s) ds\right).
\end{align}
A point process is usually defined by specifying its conditional distribution $f^*(t)$ or equivalently its intensity $\lambda^*(t)$. In the simplest case, the intervals $\delta t_i$ are assumed to be $i.i.d.$,  therefore the process is memoryless, and hence  $\lambda^*(t)=\lambda(t)$. The Cox process \cite{cox1980point} is a doubly stochastic point processes, and conditioned on the intensity is a Poisson process \cite{Kingman1992}. Hawkes process \cite{Hawkes1971} is a special type of Cox process, where the intensity is expressed by the history as:
\begin{align}
	\lambda^*(t) = \mu + \int_{-\infty}^{t} \phi(t-\tau) \, dN(\tau) 
	= \mu + \sum_{i=1}^{|\mathcal{F}_t|} \phi(t-t_i)  
\end{align}
where $\phi(t)$ is the kernel of the Hawkes process that defines the effect of past events on the current intensity, and $\mu$ is the base intensity. For example, the exponential kernel $\phi(t) = \exp(-t)$, is used for the modeling of self-exciting events like earthquake \cite{ogata1988statistical}. 
In general, we have a multivariate process with a counting process vector $\bm{N}(t) = [N_1(t),\cdots,N_n(t)]^T$ and an associated intensity vector $\bm{\lambda}^*(t) = [\lambda^*_1(t),\cdots,\lambda^*_n(t)]^T$ defined as:
 \begin{align}
 	\bm{\lambda}^*(t) = \bm{\mu} + \bm{A} \int_{-\infty}^{t} \Phi(t-\tau) \, d\bm{N}(\tau)
 \end{align} 
where $\Phi(t)$ is the matrix of mutual kernels, \textit{i.e.}, $\Phi_{ij}(t)$ models the effect of events of counting process $N_j(t)$ on $N_i(t)$, $\bm{\mu} = [\mu_1,\cdots,\mu_n]^T$ is the base intensity, and $\bm{A}=[\alpha_{ij}]$ is a matrix of mutual-excitation kernels. Often, the point process carries other information than the time of events, which is called mark. For example, the strength of an earthquake can be considered as a mark. The mark $m$, often a subset of $\mathbb{N}$ or $\mathbb{R}$, is associated with each event through the conditional mark probability function $f^*(m|t)$:
\begin{align}
	\lambda^*(t,m) = \lambda^*(t) \, f^*(m|t)
\end{align}
The mutually-exciting property of the Hawkes process makes it a common modeling tool in a variety of applications such as seismology, neurophysiology, epidemiology, reliability, and social network analysis \cite{Farajtabar2015,valera2015,Rodriguez2013,Rodriguez2011}.
\section{Problem Definition} 
%
Consider a directed network $\mathcal{G}=(\mathcal{V},\mathcal{E})$, with $\left\vert\mathcal{V}\right\vert=N$ users and $L$ locations in $C$ different categories. Each user can check-in to a location and influence her neighbors. We define a check-in as a 4-tuple $(t,u,c,l)$, which shows the time $t$ that user $u$ check-in to location $l$ with category $c$. 
We observe the sequence of all check-ins in the network $\mathcal{G}$, in the time interval $[0,T]$. The observation $\mathcal{D} = \{(t_i, u_i, c_i, l_i)\}_{i=1}^{K}$, is composed of user's check-ins where $t_i \in [0,T]$, $u_i \in  \mathcal{V}$, $c_i \in \{1,2,\ldots,C\}$ and $l_i \in \{\phi_1,\phi_2,\ldots,\phi_L\}$, that $\phi_i$ can be the id or geographical coordinate of a location.
We use the following notation for the history of check-ins of user $u$ in location $l$ with category $c$ up to time $t$:
\begin{align*}
    \mathcal{D}_{ucl}(t) &= \left \{ (t_i,u_i,c_i,l_i) \in \mathcal{D} \,\vert\, t_i<t, u_i=u, c_i=c, l_i=\phi_l \right\}
\end{align*}
Moreover, we use the dot notation to represent the union over the dotted variable, \textit{e.g.}, $\mathcal{D}_{u\Cdot\Cdot}(t)$ represents the events of user $u$, before time $t$, in any location with any category. Moreover, $\mathcal{D}_{\bar{u}c\Cdot}(t)$ represents the events of all users except $u$, before the time $t$, in any location  with category $c$.

Given this observations, we want to infer the latent influence network, and model the spatio-temporal behavior of users in the location-based social networks like Foursquare. In other words, we want to model the location and time of the next check-in of a user, by observing the history of the user and her friends. 
In this paper, we assume that users have a periodic pattern in the time of their check-ins, and are influenced by the behavior of their friends.
Therefore, we model the time of check-ins by a periodic point process 
which incorporates the periodic pattern in the users' movements, 
and the location of check-ins by a time-dependent multinomial distribution which incorporates the mutually exciting effect of friends.
\section{Proposed Method}
\begin{table}
\centering
\caption{List of symbols.}
\label{tbl:notation}
\begin{tabular}{c|l} 
	\textbf{Symbol} & \textbf{Description} \\ 
	\hline
	$\phi_l$ & Identity of the $l$'th location  \\ 
	$\beta_{u}$	& Temporal kernel parameter of user $u$ \\	
	$\mu_{uc}$	& Base temporal intensity of user $u$ in category $c$ \\
	$\alpha_{vu}$	& The influence of users $v$ on $u$  \\	
	$\eta_{uc}$	& Tendency of user $u$ to explores  new locations with category $c$ \\			
	$w_{ucl}$	& Weight of location $l$ with category $c$ for user $u$\\
	$m_{cl}$	& Overall weight of location $l$ with category $c$
\end{tabular}
\end{table}

\subsection{Modeling the Time of Check-ins}
In every working day, a user may check-in to her office in the morning then go to a restaurant at noon, and also have a weekly football practice program. 
By observing the history of the time of check-ins of a user, if she repeats some patterns recently (within several days), for example take a walk every afternoon, then it is more likely to repeat this pattern shortly in the upcoming days at  approximately the same time. It means, there is a periodicity in the users' behaviors. 
Moreover, there maybe also a drift or an addition of a new activity in the user's behavior, for example, the working hour of her office may change or there may be a new weekly social gathering. Therefore, we need a periodic point process to model the time of user's check-ins, which can also adapt to the new users' check-ins. 
This is in contrast to the self-exciting nature of the Hawkes process, which is used to model the diffusion of information over a network \cite{Rodriguez2011,Rodriguez2013,yang2013mixture}.

We propose a doubly stochastic point process which is periodic, and also has a diminishing property that enables the process to change its periodic pattern and adapt to the new behaviors. 
The proposed process, is composed of a Poisson process with the base intensity $\mu$, where each event $t_i$ of this process triggers a Poisson process with the following intensity:
\begin{align}
	\lambda_{t_i}(t) = \sum_{k=1}^{\infty} h(t-t_i-k\tau) \, g(k)
\end{align} 
where $h(t)$ is the kernel of the process, $g(k)$ is a decreasing function to diminish the intensity in the future periods, and the hyper-parameter $\tau$ is the period. This intensity is illustrated in Fig. \ref{fig:periodicHawkes}. The self-exciting property of the Hawkes process can be observed from its exponentially decaying kernel in Fig. \ref{fig:periodicHawkes}. In the Hawkes process when an event occurs, there is a high probability to have events just after it, and this probability decreases exponentially afterward. But in the proposed  process, there is a high probability to have events in the upcoming periods and this probability also decreases exponentially. 
\begin{figure}
\centering
\includegraphics[width=0.4\textwidth]{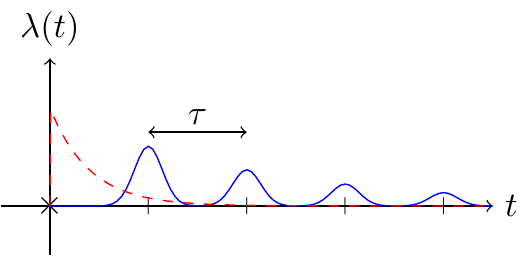}
\caption{An event at time $t=0$ triggers a poisson process. The solid curve shows the intensity of the proposed periodic point process with a Gaussian kernel and period $\tau$, and the dashed curve shows a Hawkes process with an  exponential decaying kernel.}
\label{fig:periodicHawkes}
\end{figure}	

According to the superposition theorem \cite{Kingman1992}, the intensity of the proposed process can be written as follows:
\begin{align}
	\lambda^*(t) = \mu + \sum_{i=1}^{|\mathcal{F}_t|} \lambda_{t_i}(t) = \mu + \sum_{i=1}^{|\mathcal{F}_t|} \sum_{k=1}^{\infty} h(t-t_i-k\tau) \, g(k)
\end{align}
To preserve the locality in time, the kernel $h(t)$ should have a peak at $t=0$ and decay to zero in both sides when $t \rightarrow \pm \infty$. For example, the Gaussian kernel, $h(t) = \exp(-t^2/2\sigma^2)$ meets this requirements. This model has three main features:
\begin{enumerate}
	\item \textit{Periodic Nature}. When an event occurs in time $s$, then the intensity of events around this time in the upcoming periods, $s+k\tau$, would increase.
	\item \textit{Temporal Locality}. The intensity is high around the peak of the kernel and drops rapidly in both sides.
	\item \textit{Adaptability}. The peak of the kernel decreases by the increase of $k$, 
so the process can adopt its intensity to any new periodic patterns.
\end{enumerate}
If we use a truncated Gaussian kernel like $h(t) = \exp(-t^2/2\sigma^2) \, \mathbb{I}(-\tau/2 \leq t \leq \tau/2)$, then we can substantially reduce the complexity of the intensity function. With this kernel we can show that:
\begin{align}
	\lambda^*(t) = \mu + \sum_{i=1}^{|\mathcal{F}_t|} h(t-t_i-k_i\tau) \, g(k_i)
\end{align}
where $k_i = \lfloor \frac{t-t_i}{\tau} \rfloor$ is the period number of which the event in $t_i$ affects on the current intensity. 
We propose the following point process for the time of check-ins of  user $u$ in any location with category $c$:
\begin{align}\label{equ:spt-t-m}
	\lambda_{u}(t,c) = \mu_{uc} +  \sum_{i=1}^{\vert\mathcal{D}_{uc\Cdot}(t)\vert}  \beta_{u} \exp\left[-\frac{(t-t_i-k_i\tau)^2}{2\sigma^2}\right] \exp(-k_i)
\end{align}
The first term, $\mu_{uc}$ is the base intensity that models the external effect on user $u$ to generates check-ins with category $c$, the second term is the periodic effect of the history, $\beta_u$ is the kernel parameter, and $\tau,\,\sigma $ are hyper-parameters. All parameters of the model are listed in Table \ref{tbl:notation}.
The intuition of this model is that, if a user check-ins frequently, for example in the ''restaurant`` category at noon, then with high probability, she will checks in a restaurant at noon in the next day. 

\subsection{Modeling the Location of Check-ins}
In this section, we propose a model for the location of users' check-ins, given the history of check-ins. 
We use the fact that, users in social networks are influenced by the behavior of their neighbors. 
Let denote the weight of location $l$ with category $c$ for user $u$ as:
\begin{align}\label{equ:spt-weight-w}
w_{ucl} = \sum_{i=1}^{\vert\mathcal{D}_{\Cdot cl}(t)\vert} \alpha_{u_iu} \exp(-(t-t_i))
\end{align}
which incorporates $\alpha_{u_iu}$, the influence of user $u_i$ on $u$, and the time of check-ins with an exponentially decaying kernel. This kernel diminishes the effect of far past check-ins, so the model can adopt to any new behaviors of the users' check-ins. Therefore, a location which checked in recently with many or even few but influential friends would have high weight. We also define a weight for the popularity of a location $l$ with category $c$ from the perspective of all users: 
\begin{align}\label{equ:spt-weight-m}
m_{cl} = \sum_{i=1}^{\vert\mathcal{D}_{\Cdot c l}(t)\vert} \exp(-(t-t_i))
\end{align} 
where the location that is most checked in recently, has the highest weight. 

When a user decides to check-in for example, at a restaurant, she selects a location that herself or her friends have checked in frequently, recently (exploitation effect), and sometimes she check-ins to a new popular restaurant (exploration effect).
Therefore, we use the following multinomial conditional distribution to define the probability that user $u$ check-ins to location $\ell$, given the time $t$ and category $c$:
\begin{align}\label{equ:spt-spatio}
f_u(\ell|c,t)
&= \underbrace{\sum_{l=1}^{L} \frac{w_{ucl}}{\eta_{uc} + w_{uc\Cdot} } \delta_{\phi_{l}}(\ell)}_{\text{exploitation}} + \underbrace{\frac{\eta_{uc}} {\eta_{uc} + w_{uc\Cdot}}  G_0(\ell)}_{\text{exploration}}
\end{align}
The Dirac delta function $\delta_{\phi_l}(\ell)$ is $1$ if $\phi_l=\ell$, otherwise it is $0$, and the parameter $\eta_{uc}$ models the inclination of the user to explores new locations. This distribution means that, with probability ${w_{ucl}}/(\eta_{uc} + w_{uc\Cdot})$ the current location would be a previously checked in location $\phi_l$ by the user $u$ or any of her friends (since for non visited locations the weight $w_{ucl}$ is zero), and with probability ${\eta_{uc}}/(\eta_{uc} + w_{uc\Cdot})$ it would be selected from all locations in the network, with a probability that is modeled by the following distribution:
\begin{align}
	G_0(\ell) &= \sum_{l=1}^{L} \frac{m_{cl}}{m_{c\Cdot}} \delta_{\phi_l}(\ell)	
\end{align}
Where according to the definition of coefficient $m_{cl}$, it assigns more probability to the popular or recently frequently visited locations.
The main features of the proposed location model are:
\begin{enumerate}
	\item \textit{Exploitation}. The future check-ins of a user are influenced by the history of check-ins of the user and her friends.
	\item \textit{Exploration}. There is a probability that users explore and check into new unseen locations.
	\item \textit{Adaptability}. Using exponential decaying kernel for the weights, the model can adopt to new patterns in users' behavior.
\end{enumerate}
\begin{algorithm}[t]
\small
\DontPrintSemicolon 
\KwIn{$N,\,C,\,L$, all parameters $\{\mu_{uc}, \eta_{uc}, \alpha_{uv},\beta_{u}\}$, history of check-ins.}
\KwOut{Next check-in $(t_i,u_i,c_i,l_i)$.}
\For{$u=1:N$}{$\lambda_u(t) = \sum_c \lambda_u(t,c)$}
$\lambda(t) = \sum_{u} \lambda_u(t) $\;
$t_i \sim \mathcal{PP}(\lambda(t))$\;
$u_i \sim \text{Multi}(\frac{\lambda_1(t_i)}{\lambda(t_i)},\ldots,\frac{\lambda_N(t_i)}{\lambda(t_i)}) $\;
$c_i \sim \text{Multi}(\frac{\lambda_{u_i}(t_i,1)}{\lambda_{u_i}(t_i)},\ldots,\frac{\lambda_{u_i}(t_i,C)}{\lambda_{u_i}(t_i)}) $\;
$l_i \sim f_{u_i}(\ell \vert c_i,t_i)$\;
\Return{$(t_i,u_i,c_i,l_i)$}\;
\caption{Generative model of the check-ins.}
\label{alg:sampling}
\end{algorithm}

\subsection{Summary of Generative Model}
The proposed generative model is summarized in Algorithm \ref{alg:sampling}. Using the superposition theorem, first the time $t$  of  check-in is sampled from the proposed periodic point process $\lambda(t)=\sum_{u,c} \lambda_u(t,c)$, then the user $u$ which generated this event is selected in proportion to its intensity $\lambda_u(t)$. The category $c$ of the check-in is also selected in proportion to $\lambda_u(t,c)$. Finally, the location $l$ is sampled from the proposed location model.  

\subsection{Inference}
We propose a Bayesian inference algorithm based on the EM algorithm to find the model parameters. 
To find the maximum likelihood solution, for each check-in $(t_i,u_i,c_i,l_i)$, we define a latent variable $z_i$ as the user that caused $u_i$ to check into location $l_i$, given the time $t_i$ and category $c_i$. We use $1$-of-$N$ coding to represent $z_i$'s. For notional convenient, lets define: 
\begin{align}\label{equ:gamma}
	\gamma_{u c \ell}^{v} &= \frac{w_{u c \ell}^{v}} {\eta_{u c} + w_{u c\Cdot}} \mathbb{I}(v>0) +
	\frac{m_{c\ell} \: \eta_{u c}} {m_{c\Cdot}(\eta_{u c} + w_{u c\Cdot})} \mathbb{I}(v=0) \\
	w_{ucl}^v &= \sum_{i=1}^{\vert\mathcal{D}_{vcl}(t)\vert} \alpha_{vu} \exp({-(t-t_i)}) = 
	\alpha_{vu} \sum_{i=1}^{\vert\mathcal{D}_{vcl}(t)\vert} \exp({-(t-t_i)})
\end{align}
where $\gamma_{u c \ell}^{v}$ is the contribution or influence of user $v$ in the check-in of user $u$ at location $l$ with category $c$. Now, we define:
\begin{align}\label{equ:mark-prob}
	f_{u_i}(l_i,z_i \vert t_i, c_i) = \prod_{v=0}^N (\gamma_{u_i c_i \ell_i}^{v})^{z_{iv}}	
\end{align}
where $z_{iv}$ is the $v$'th element of $z_i$, or the index of the user that caused $i$'th check-ins. But, $v=0$ is not the index of  a user, it represents the exploration effect.
It can be verified that marginalizing out the $z_i$, $ \sum_{z_i} f_{u_i}(l_i,z_i \vert t_i, c_i)$, results in the probability distribution (\ref{equ:spt-spatio}).
Now, to evaluate the complete likelihood $p(\mathcal{D}, Z \vert \theta)$ of the data $\mathcal{D}$ and hidden variables $Z=\{z_i\}_{i=1}^{K}$, given the parameters $\theta = \{\mu_{uc}, \eta_{uc}, \alpha_{uv},\beta_{u}\}$,  $u,v=1\ldots N$ and $c=1\ldots C$,  we use the following proposition.

\begin{proposition}[\cite{zarezade2015correlated}]
\label{prop1}
Let $N_u, \; u=1,2,\cdots,N$ be a multivariate marked point process with the associated intensity $\lambda_u(t)$, and the mark probability $f_u(m|t)$. Let 
$\mathcal{D}=\{(t_i,u_i,m_i)\}_{i=1}^K$ be a realization of the process over $[0,T]$. Then the likelihood of $\mathcal{D}$ on model $N_u$ with parameters $\theta$ can be expressed as follows.
\begin{align*}
p(\mathcal{D} \vert \theta) =\exp\left(-\int_0^T\sum_{u=1}^N \lambda_u(\tau) \; d\tau \right)
 \prod_{i=1}^{\vert\mathcal{D}\vert} \lambda_{u_i}(t_i) f_{u_i}(m_i|t_i) 
\end{align*}
\end{proposition}

If we consider $(c_i,l_i,z_i)$ as the mark $m_i$ of the process, 
according to this proposition the complete likelihood of our model is,
\begin{align}
	p(\mathcal{D}, Z \vert \theta) &= \exp\left(-\int_0^T\sum_{u=1}^N \lambda_u(\tau) \; d\tau \right) \prod_{i=1}^{\vert\mathcal{D}\vert} \lambda_{u_i}(t_i) f_{u_i}(c_i, l_i, z_i|t_i)
\end{align}
where using Bayes' rule and Eq. (\ref{equ:mark-prob}) it can be evaluated as follows.
\begin{align}
	p(\mathcal{D}, Z \vert \theta) &= \exp\left(-\int_0^T\sum_{u=1}^N \lambda_u(\tau) \; d\tau \right) \prod_{i=1}^{\vert\mathcal{D}\vert} \lambda_{u_i}(t_i) f_{u_i}(c_i|t_i)f_{u_i}(l_i, z_i|t_i, c_i) \nonumber \\
	&= \exp\left(-\sum_{u=1}^N \sum_{c=1}^C \int_0^T \lambda_{u}(\tau,c) \; d\tau \right) \prod_{i=1}^{\vert\mathcal{D}\vert} \lambda_{u_i}(t_i,c_i) f_{u_i}(l_i, z_i|t_i, c_i) \nonumber \\
	&= \exp\left(-\sum_{u=1}^N \sum_{c=1}^C \int_0^T \lambda_{u}(\tau,c) \; d\tau \right) \prod_{i=1}^{\vert\mathcal{D}\vert} \lambda_{u_i}(t_i, c_i) \prod_{v=0}^N (\gamma_{u_i c_i l_i}^{v})^{z_{iv}} \nonumber 
\end{align}
To derive the second line, we used the superposition theorem, and the fact that the probability of a category, according to our generative model is $f_{u_i}(c_i|t_i)=\lambda_{u_i}(t_i,c_i)/\lambda_{u_i}(t_i)$.
Given the joint distribution of the observed and latent variables $p(\mathcal{D},Z \vert \theta)$, we use EM algorithm to maximize the likelihood function $p(\mathcal{D} \vert \theta)$ with respect to $\theta$. 
In the E-step we evaluate $p(Z \vert \mathcal{D}, \theta)$. 
Using Bayes' rule we can write the posterior distribution of the  latent variables as,
\begin{align}
	p(Z \vert \mathcal{D}, \theta) \propto
\prod_{i=1}^{\vert\mathcal{D}\vert}  \prod_{v=0}^N (\gamma_{u_i c_i l_i}^{v})^{z_{iv}}	
\end{align}
which factorizes over $i$, so that $z_i$'s are independent with multinomial distribution and we can write the expected of $z_{iv}$ under this distribution as follows.
\begin{align}
	\mathbb{E}[z_{iv}] = \frac{ \sum_{z_{iv}} z_{iv} (\gamma_{u_i c_i \ell_i}^{v})^{z_{iv}} }{\sum_{z_i}\prod_{v=0}^{N} (\gamma_{u_i c_i \ell_i}^v)^{z_{iv}}}
	= \frac{\gamma_{u_i c_i \ell_i}^{v}}{\sum_{v=0}^{N} \gamma_{u_i c_i \ell_i}^v}
\end{align}
In the M-step we maximize $\mathbb{E}_Z[\ln p(\mathcal{D}, Z \vert \theta)]$ the expected complete log-likelihood, which can be decomposed to the sum of expected log-likelihoods of users $\mathbb{E}_{Z_u}[\ln p(\mathcal{D}_u, Z_u \vert \theta_u)]$.
\begin{align}
	\mathbb{E}_Z[\ln p(\mathcal{D}, Z \vert \theta)] =& - \sum_{u=1}^N \sum_{c=1}^C \int_0^T \lambda_{u}(\tau,c) \; d\tau + \sum_{i=1}^{\vert\mathcal{D}\vert} \log\lambda_{u_i}(t_i,c_i) 
	+ \sum_{i=1}^{\vert\mathcal{D}\vert} \sum_{v=0}^N \mathbb{E}[z_{iv}] \log \gamma_{u_i c_i \ell_i}^{v} \nonumber \\
	&= \sum_u \left( -\int_0^T \sum_{c=1}^C \lambda_{u}(\tau,c) \; d\tau + \sum_{i=1}^{\vert\mathcal{D}_u\vert} \log\lambda_{u}(t_i,c_i) + \sum_{i=1}^{\vert\mathcal{D}_u\vert} \sum_{v=0}^N \mathbb{E}[z_{iv}] \log \gamma_{u c_i \ell_i}^{v} \right)  \nonumber \\
	&= \sum_u \mathbb{E}_{Z_u}[\ln p(\mathcal{D}_u, Z_u \vert \theta_u)] \label{equ:users-lglk}
\end{align}
Where $Z_u = \{z_i \in Z \,\vert\, u_i=u\}$ and $\theta_u = \{\mu_{uc}, \eta_{uc}, \alpha_{uv},\beta_{u}\}$,  $v=1\cdots N$, $c=1\cdots C$.
Accordingly, the M-step can be decomposed to multiple maximizations over users, which can be done in parallel. The two steps of the EM algorithm can be summarized as follows.
\begin{align*}
	&\text{E-Step:}\quad  \mathbb{E}[z_{iv}] =  \frac{\gamma_{u_i c_i \ell_i}^{v}}{\sum_{v=0}^{N} \gamma_{u_i c_i \ell_i}^v}
 \\
	&\text{M-Step:}\quad \theta^*_u = \argmax_{\theta_u \geq 0} \; -\int_0^T \sum_{c=1}^C \lambda_{uc}(\tau) \; d\tau + \sum_{i=1}^{\vert\mathcal{D}_u\vert} \log\lambda_{u c_i}(t_i) + \sum_{i=1}^{\vert\mathcal{D}_u\vert} \sum_{v=0}^N \mathbb{E}[z_{iv}] \log \gamma_{u c_i l_i}^{v}
\end{align*}
In the following proposition, we prove that the maximization in M-step is concave, so it has a unique and optimal solution. 

\begin{proposition}
The expected log-likelihood of a user, $\mathbb{E}_{Z_u}[\ln p(\mathcal{D}_u, Z_u \vert \theta_u)]$ as a function of $\{\mu_{uc}, \tilde{\eta}_{uc}, \tilde{\alpha}_{uv},\beta_{u}\}$  is concave, where $\alpha_{uv}=\exp({\tilde{\alpha}_{uv}})$ and $\eta_{uc} = \exp({\tilde{\eta}_{uc}})$.
\end{proposition}
\begin{proof}.
According to Eq. (\ref{equ:users-lglk}) the log-likelihood of user $u$ is:
\begin{align}
\mathbb{E}_{Z_u}[\ln p(\mathcal{D}_u, Z_u \vert \theta_u)] 
&= -\int_0^T \sum_{c=1}^C \lambda_{u}(\tau,c) \; d\tau + \sum_{i=1}^{\vert\mathcal{D}_u\vert} \log\lambda_{u}(t_i, c_i) + \sum_{i=1}^{\vert\mathcal{D}_u\vert} \sum_{v=0}^N \mathbb{E}[z_{iv}] \log \gamma_{u c_i \ell_i}^{v}  \nonumber
\end{align}
The first term is a linear function of $\{\mu_{uc},\beta_{u}\}$, so it is both convex and concave. The second term is the log of a linear function which is concave, according to composition rules \cite{Boyd2004}. The third term is composed of $\log \gamma_{u c_i l_i}^{v}$, which for $v>0$, 
\begin{align*}
	\log \gamma_{u c_i \ell_i}^{v} 
&=\tilde{\alpha}_{vu} - \log\Big(e^{\tilde{\eta}_{u c_i}} + \sum_{j=1}^{\vert\mathcal{D}_{\Cdot c_i\Cdot}(t)\vert} e^{\tilde{\alpha}_{u_j u}} e^{-(t-t_j)}\Big) + \text{const}
\end{align*}
and for $v=0$,
\begin{align*}
\log \gamma_{u c_i \ell_i}^{v}
&=\tilde{\eta}_{uc_i} - \log\Big(e^{\tilde{\eta}_{u c_i}} + \sum_{j=1}^{\vert\mathcal{D}_{\Cdot c_i\Cdot}(t)\vert} e^{\tilde{\alpha}_{u_j u}} e^{-(t-t_j)}\Big) + \text{const}.
\end{align*}
In both cases $\log \gamma_{u c_i l_i}^{v}$ is concave according to Lemma 1 of \cite{zarezade2015correlated} which state that logarithm of sum of linear exponentials is convex. So, the overall expression is concave. Actually, we use $\tilde{\eta}_{uc}, \tilde{\alpha}_{uv}$ instead of ${\eta}_{uc}, {\alpha}_{uv}$ in the implementations, and solve the resulting concave optimization.
\end{proof}




\section{Experiments}
In this section, using both synthetic and real data, we evaluate the performance of the proposed method. 

%
\begin{figure}[t]
	\centering
	\includegraphics[width=0.35\textwidth]{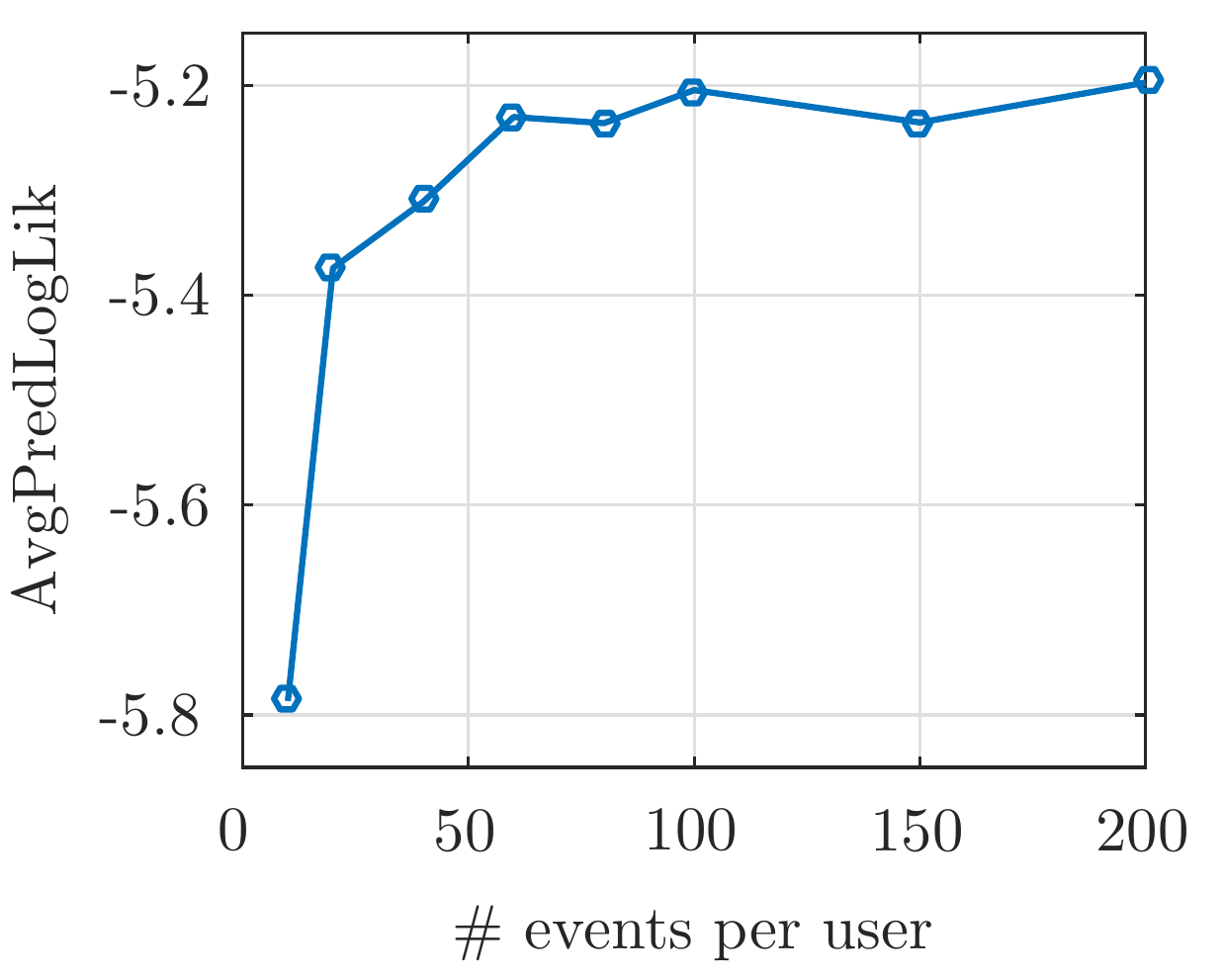}
	\hspace{0.5cm}
	\includegraphics[width=0.35\textwidth]{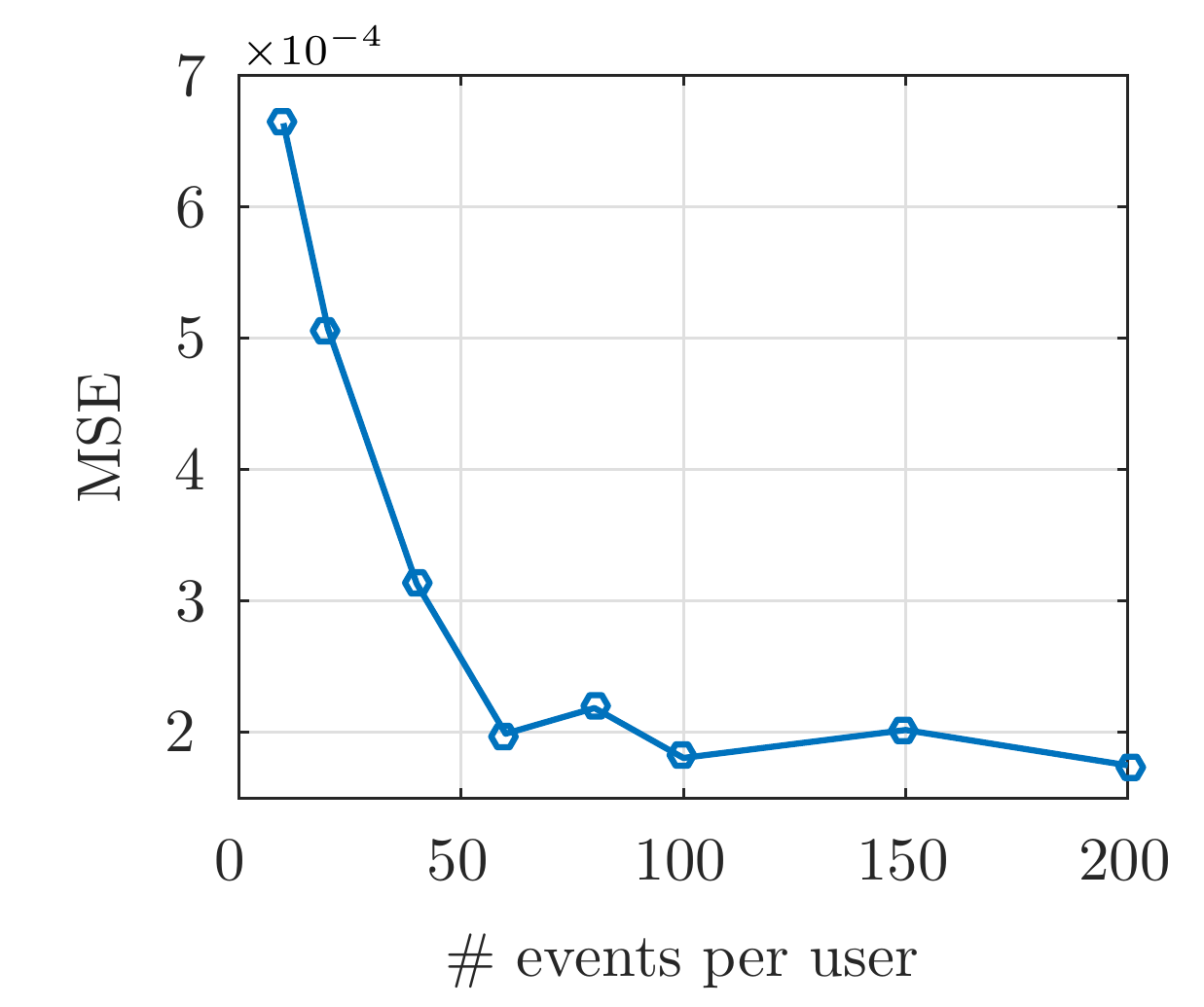}
	\caption{Average predicted log-likelihood on the test data (\textit{left}), and MSE of the learned parameters (\textit{right}), in the temporal model for the different percentages of the train data.}
	\label{fig:synth-temporal-lglk-mse}
\end{figure}
\begin{figure}[t]
	\centering
	\includegraphics[height=4.5mm]{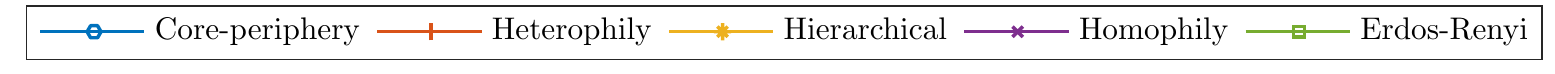}\\
	\vspace{-1mm}
	\includegraphics[width=0.35\textwidth]{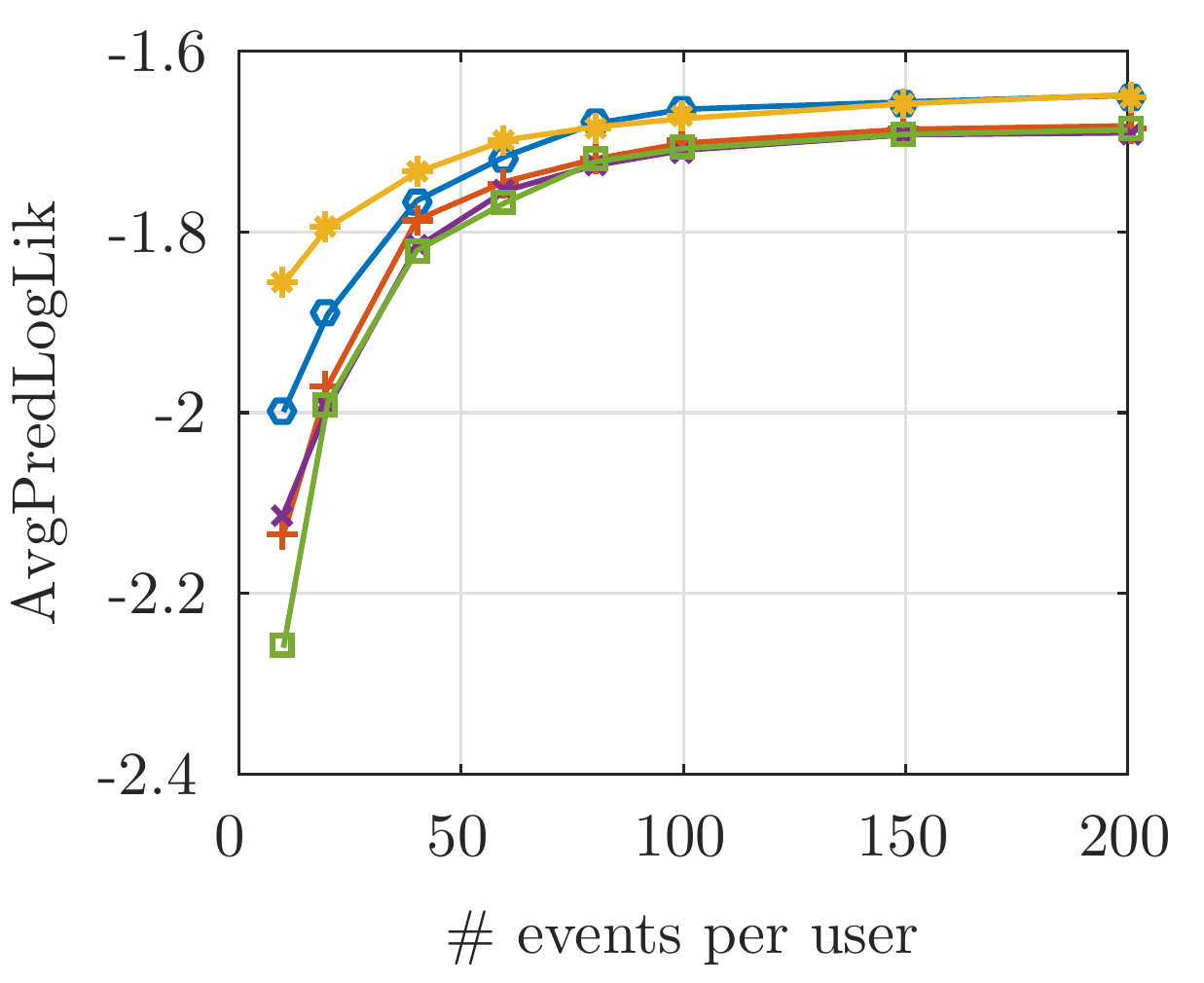}
	\hspace{0.5cm}
	\includegraphics[width=0.35\textwidth]{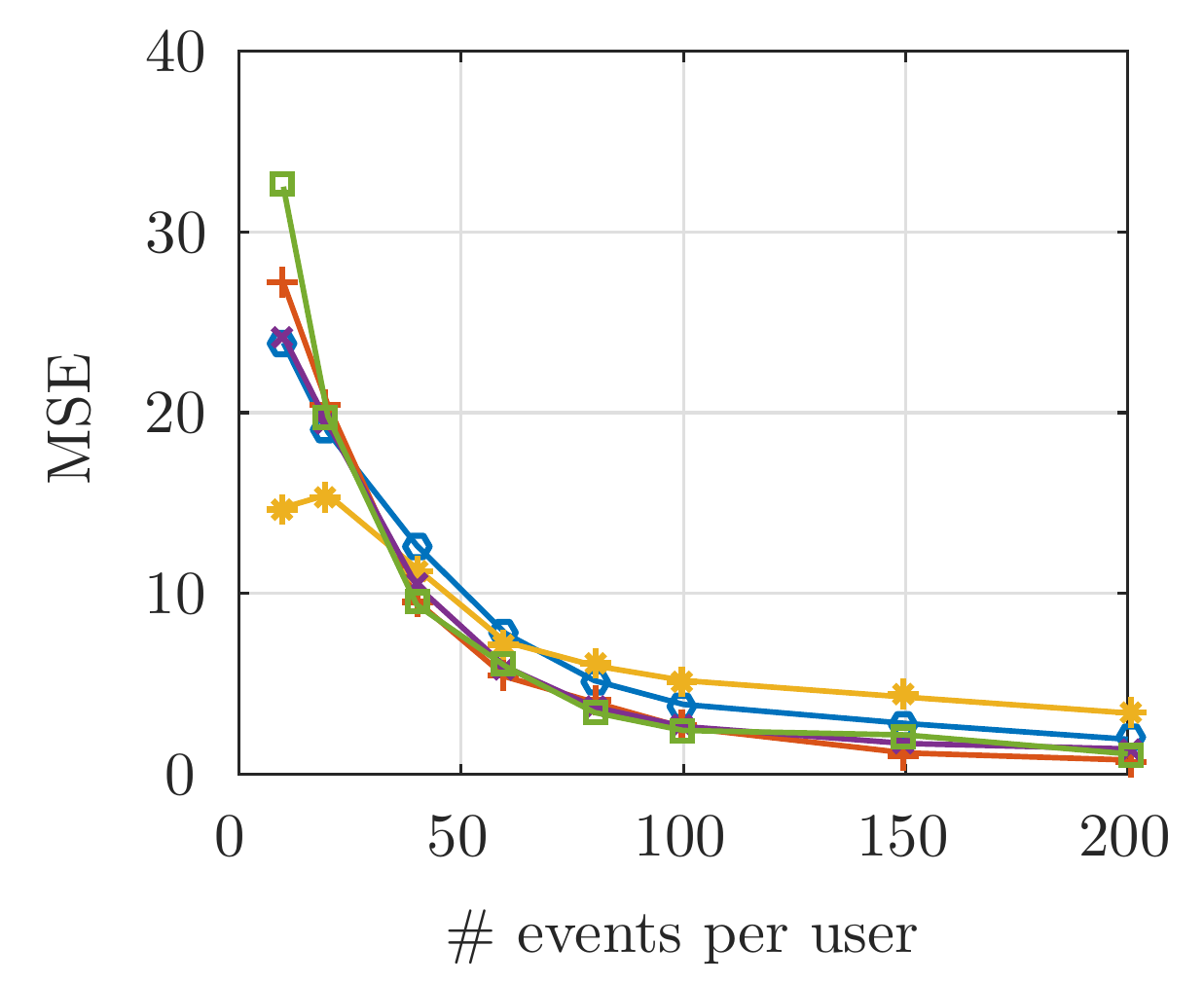}
	\caption{Average predicted log-likelihood on the test data (\textit{left}), and MSE of the learned parameters (\textit{right}), in the spatial model for the different percentages of the train data and various random graph structures.}
	\label{fig:synth-spatial-lglk-mse}
\end{figure}
%
\subsection{Experiments on Synthetic Data}
\subsubsection{Experimental setup}
We experiment with five random Kronecker networks \cite{Leskovec2010} with $N=64$ nodes, namely Core-periphery, Heterophily, Hierarchical, Homophily, and Erdos-Renyi\footnote{Where the seed matrix parameters are $[0.85,0.45;0.45,0.3]$, $[0.3,0.89;0.89,0.3]$, $[0.9,0.1;0.1,0.9]$, $[0.89,0.3;0.3,0.89]$, and $[0.60,0.60;0.60,0.60]$, respectively}. We set the number of categories to $C=4$ and consider eight locations in each category. The temporal and spatial model parameters are randomly drawn from the uniform distributions $\mu_{uc}, \eta_{uc} \sim U(0,0.05)$, $\alpha_{uv} \sim U(0,0.5)$ and $\beta_{u} \sim U(0,0.1)$. The period and standard deviation in the temporal model are $\tau=12$ and $\sigma=0.5$, respectively. We generate $16000$ check-ins from our model, using the Ogata method \cite{Ogata1981}, and consider the first $80\%$ of them for the train and the remaining $20\%$ for the test data. Then we learn the model with different percentages of the training data, and evaluate the average predicted log-likelihood on the test data ({\textit{AvgPredLogLik}) and the mean squared error between the estimated and real parameters ({\textit{MSE}). The inference   algorithm is implemented in parallel for users. All source codes and datasets are available in our git repository\footnote{{https://github.com/azarezade/stp}}.

\begin{figure}[t]
	\centering
	\includegraphics[width=0.325\textwidth]{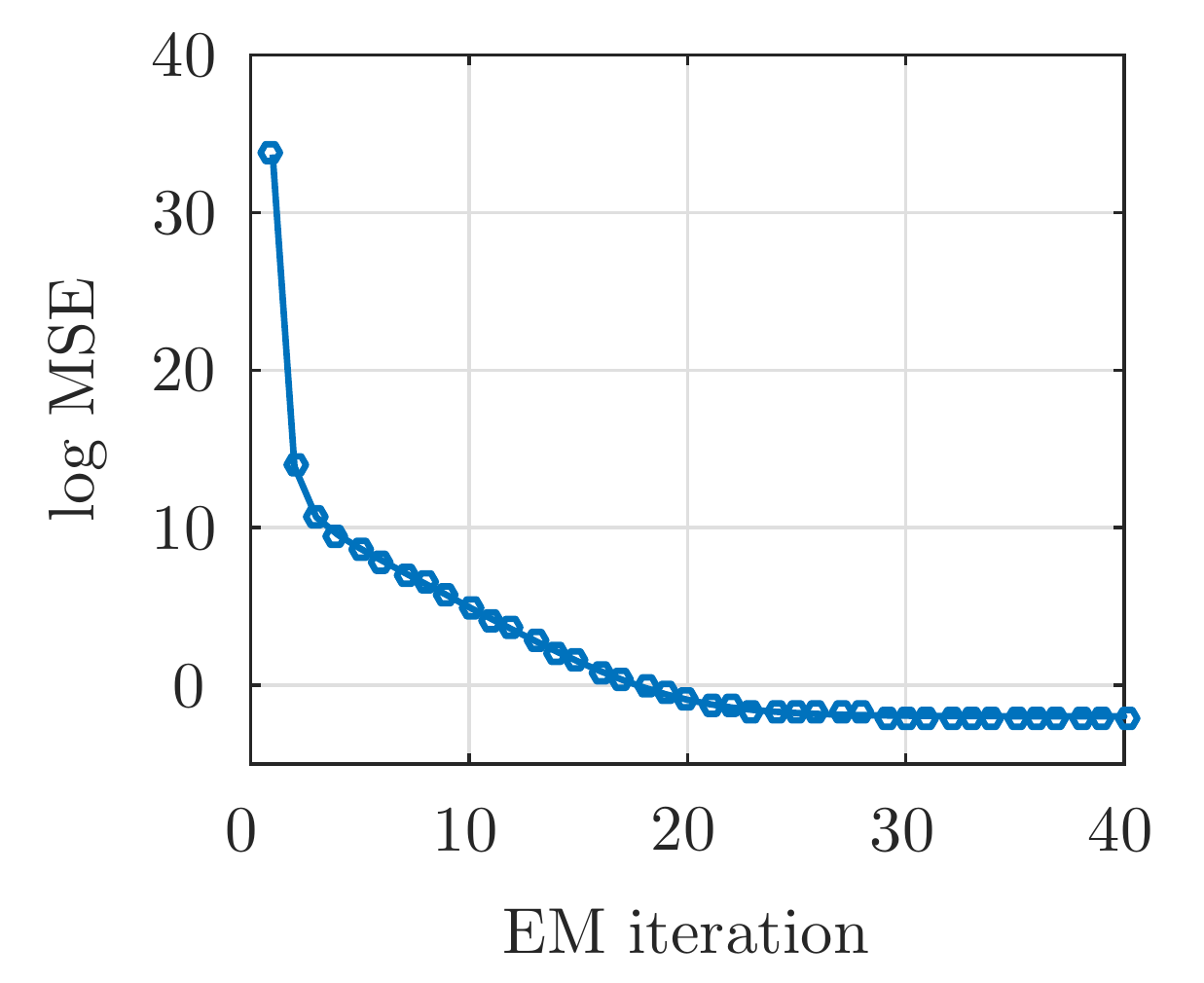}
	\includegraphics[width=0.325\textwidth]{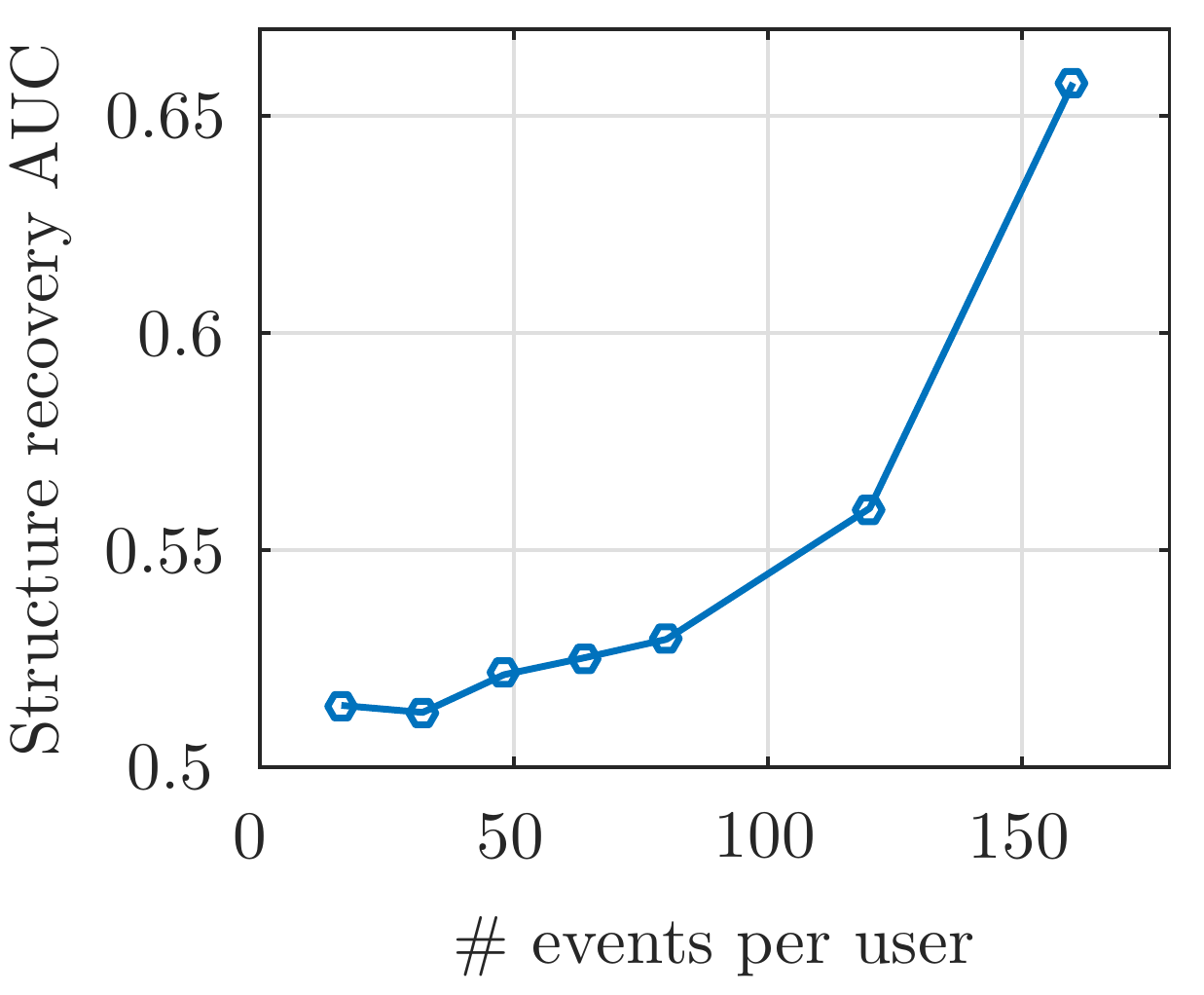}
	\includegraphics[width=0.325\textwidth]{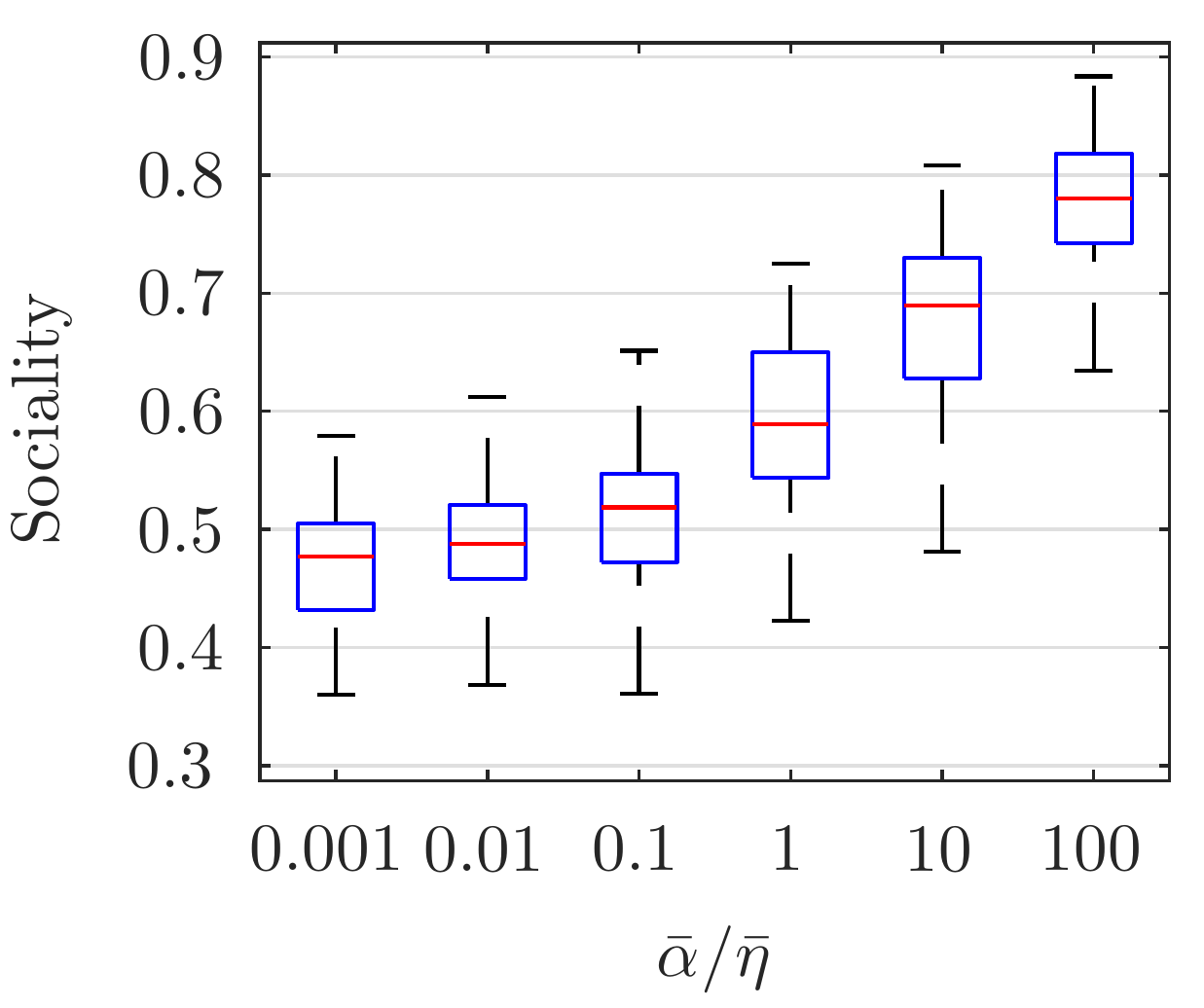}
	\caption{Average predicted log-likelihood in logarithmic scale vs the iterations of EM (\textit{left}), the network structure recovery for different percentages of the train data (\textit{middle}), and the effect of spatial parameters on the users' Sociality (\textit{right}).}
	\label{fig:synth-em-auc}
\end{figure}
%
\subsubsection{Results}
In Fig. \ref{fig:synth-temporal-lglk-mse} the {\textit{AvgPredLogLik}} and {\textit{MSE}} of the temporal model is plotted versus the size of train data, where the average estimation error decreases to about $7{\times}10^{-4}$. These measures are also plotted for the spatial model with different random network structures in Fig. \ref{fig:synth-spatial-lglk-mse}, given the time of check-ins. We can see that the parameter estimation error decreases and the average log-likelihood increases as we increase the size of train data, which shows the proposed inference algorithm can consistently learn the model parameters with a very small estimation error. Furthermore, in the left of Fig. \ref{fig:synth-em-auc} we show that for a fixed number of events per user, increasing the EM iterations would decrease \textit{MSE} to about $0.1$. 
To investigate the network structure prediction of our model, for each size of the train data, we use a threshold to convert the predicted weighted network (\textit{i.e.}, the $\alpha_{ij}$'s) to a $(0,1)$-adjacency matrix and evaluate the percent of recovered edges to form the ROC curve. Then, we find the AUC curve, which is illustrated in the middle of Fig. \ref{fig:synth-em-auc}. 

\begin{figure}[t]
	\centering
	\includegraphics[width=0.51\textwidth]{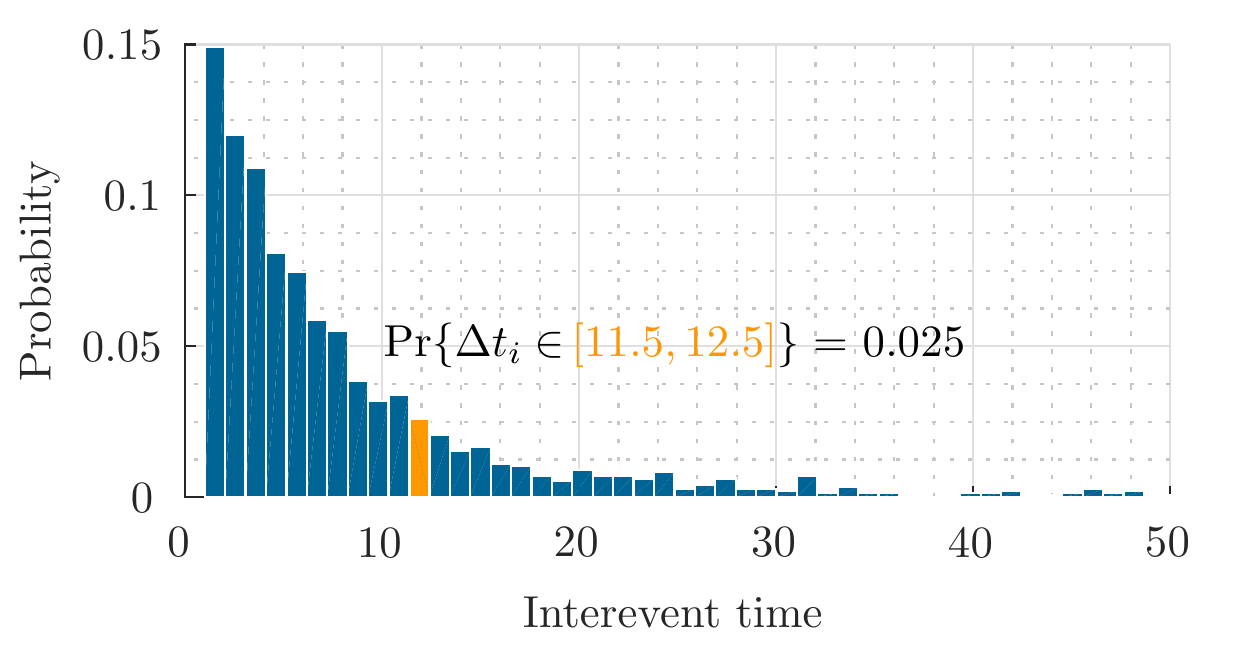}
	\hspace{-5mm}
	\includegraphics[width=0.51\textwidth]{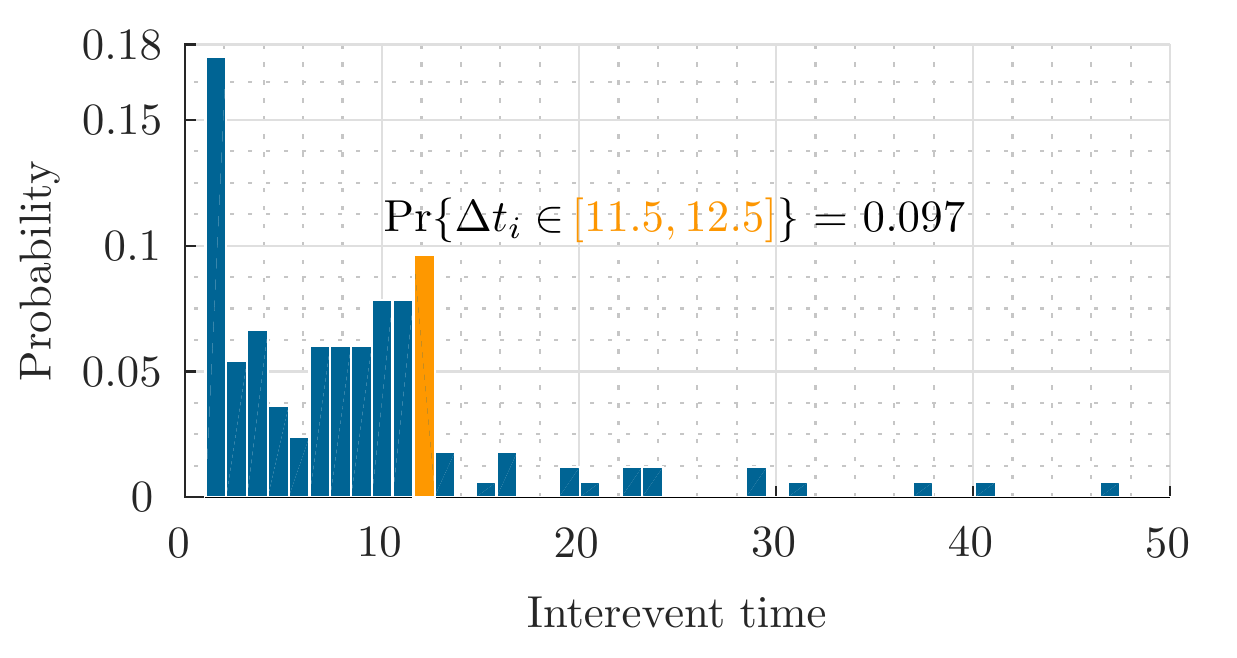}
	\caption{The distribution of interevent in the temporal model with $\beta=0$ (\textit{left}) and $\beta=1$ (\textit{right}). We can see that increasing $\beta$ would cause a peak around $12$, which is the period of the simulated events.}
	\label{fig:synth-temporal-period}
\end{figure}
\begin{figure}[t]
	\centering
	\includegraphics[width=0.59\textwidth]{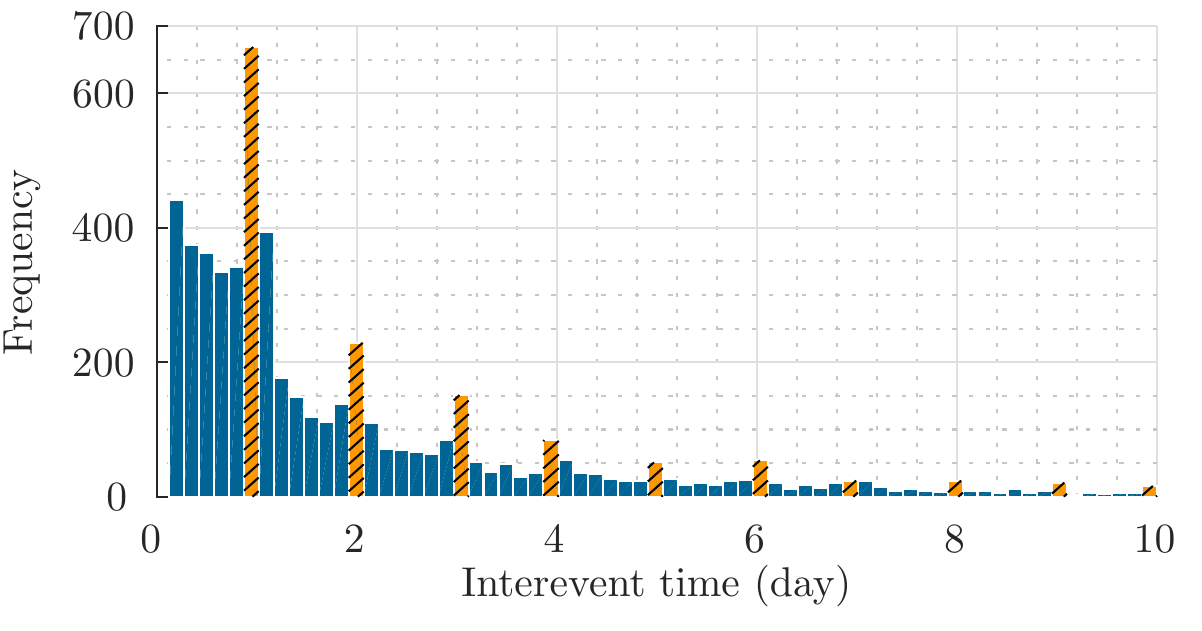}
	\includegraphics[width=0.35\textwidth]{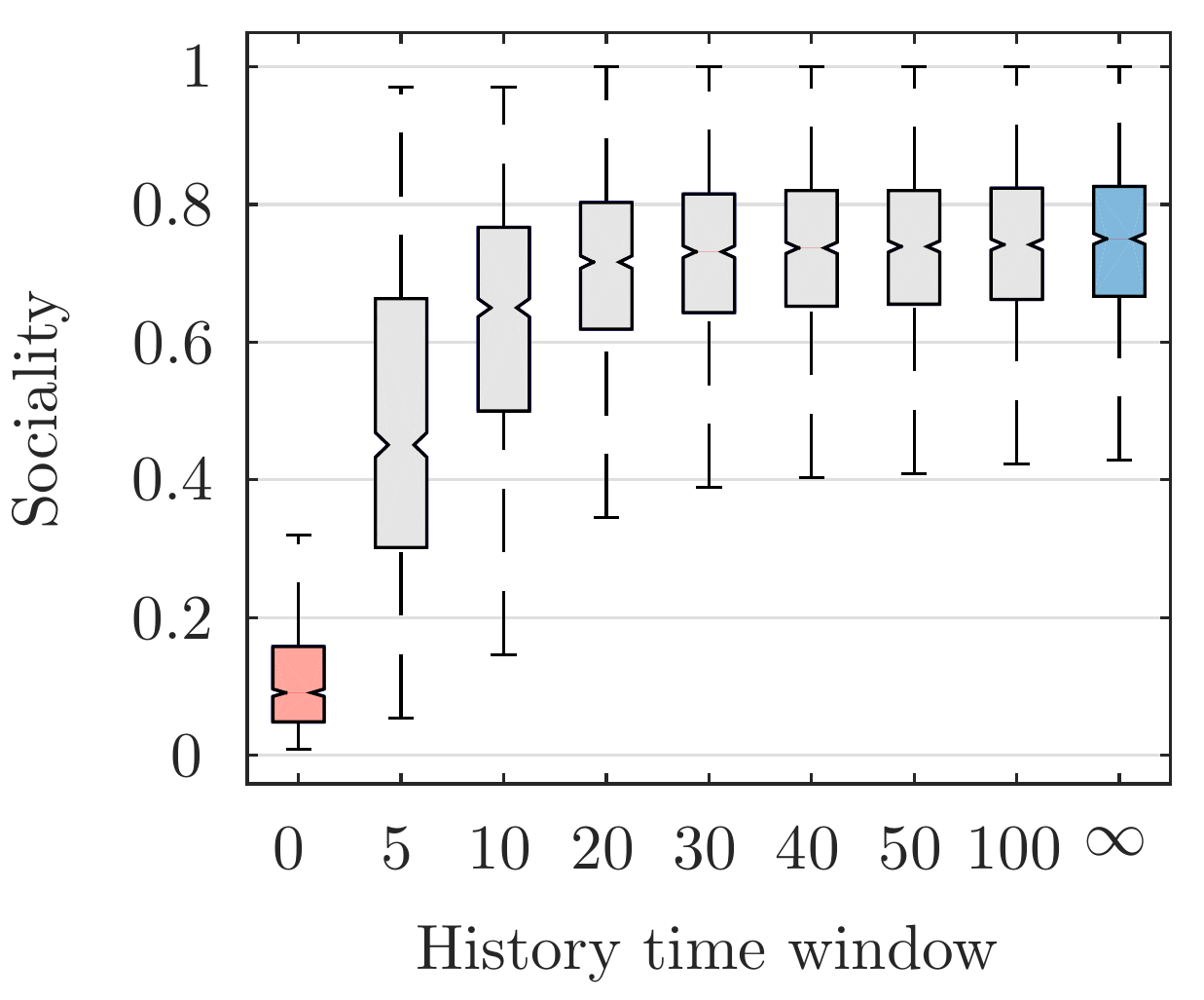}
	\caption{The frequency of interevent times in the Food category of Foursquare dataset (\textit{left}), and the Sociality box plot of users for different history window sizes (\textit{right}).}
	\label{fig:real-empirical}
\end{figure}

To study the effect of model parameters on the users' behavior, we design two experiments. First, we define a measure called \textit{Sociality}. For each user, the \textit{Sociality} is the percent of check-ins that their location has been previously visited by the user or her friends.
According to our spatial model, Eq. (\ref{equ:spt-spatio}), the exploration of users increase as we increase $\eta$ or decrease $\alpha$. To empirically validate this property of our model, in the right of Fig. \ref{fig:synth-spatial-lglk-mse} the box plot of the users' \textit{Sociality} is illustrated for different parameters. Its average reaches up to $80\%$ when the average ratio of spatial parameters, $\bar{\alpha}/\bar{\eta}$ is equal to $100$. It means that, users with high $\alpha/\eta$ are more affected by their friends. 
Moreover, to see the effect of temporal model parameters on the check-ins time of users, we plot the distribution of users' interevent time (the time difference between two successive events in a specific category for each user). According to Eq. (\ref{equ:spt-t-m}), parameters $\beta$ and $\mu$ regulate the periodicity in the time of events. The higher $\beta$, would result in more periodic events. We fix $\mu$ and set $\beta=0$ and $1$ in the left and right graphs of Fig. \ref{fig:synth-temporal-period}, respectively. As we see, there is a peak around $12$ in the right graph, which is the period of the simulated events but, in the left figure the frequency of events reduces exponentially and there is no peak except the initial one.
\subsection{Experiments on Real Data}
\subsubsection{Dataset preparation}
We used both Twitter and Foursquare APIs to crawl the check-ins data of the users in Foursquare, because Foursquare does not provide the check-ins data. Specifically, we crawled the tweets of the users that have installed Swarm application. This app is connected to the Twitter and Foursquare account of the user. When a user check-ins, using this app, she can tweet the URL of that location in the Foursquare website. Therefore, we have access to the location details (via Foursquare API) and the time of check-ins (via Twitter API).
Using the Twitter search API we found active users with high check-ins rate in Foursquare. By querying the API with ``I am at'', the default template of Swarm app for check-ins, we selected the top $12000$ users, and crawled their tweets in ten weeks during the year $2015$.
We pruned the data by selecting $1000$ active users that were in the same country (Brazil), to better see the influence of users on each other.
The average degree of the network is $6.4$. The total number of check-ins is about $60000$.
The number of unique locations is about $10000$ in $10$ categories.

\subsubsection{Experimental setup}
We use the first eight weeks of the check-ins as the train data, and the remaining two weeks as the test data. The hyper-parameters of the temporal model are set to $\tau=24$ and $\sigma=1$, by cross validation. 
We learn model parameters by the train data and use different temporal and spatial measures for the evaluations.
We compare our proposed model with \textit{MultiHawkes} \cite{yang2013mixture} where the intensity of user's check-ins is modeled by a multivariate Hawkes process (the intensity depends on the user and her friends' history), and \textit{Hawkes} where the intensity, that is modeled by a Hawkes process, depends only on the user's history.
The spatial model is also compared with two competing methods. In the \textit{MostPopular} method the most checked in locations, disregarding the time of check-ins, are more probable to be selected as the next check-in location. The \textit{PeriodicLoc} model assumes periodicity in the location of check-ins, the locations that are more checked in previous periods are more probable to be visited in the current time.

\begin{figure}[t]
	\centering
	\includegraphics[height=4mm]{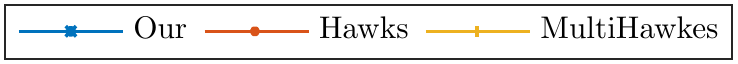}\\
	\vspace{-0.5mm}	
	\includegraphics[width=0.325\textwidth]{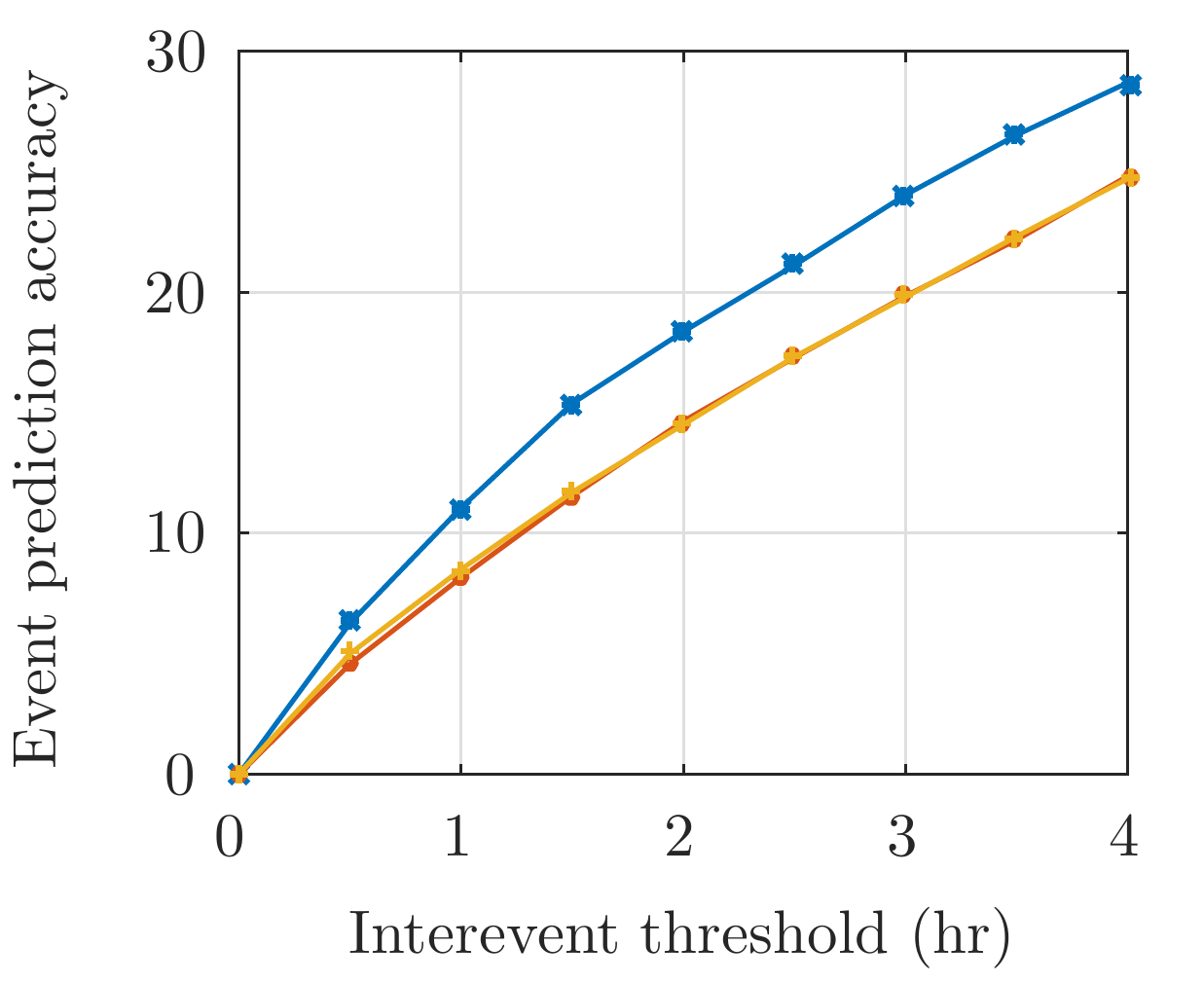}
	\includegraphics[width=0.325\textwidth]{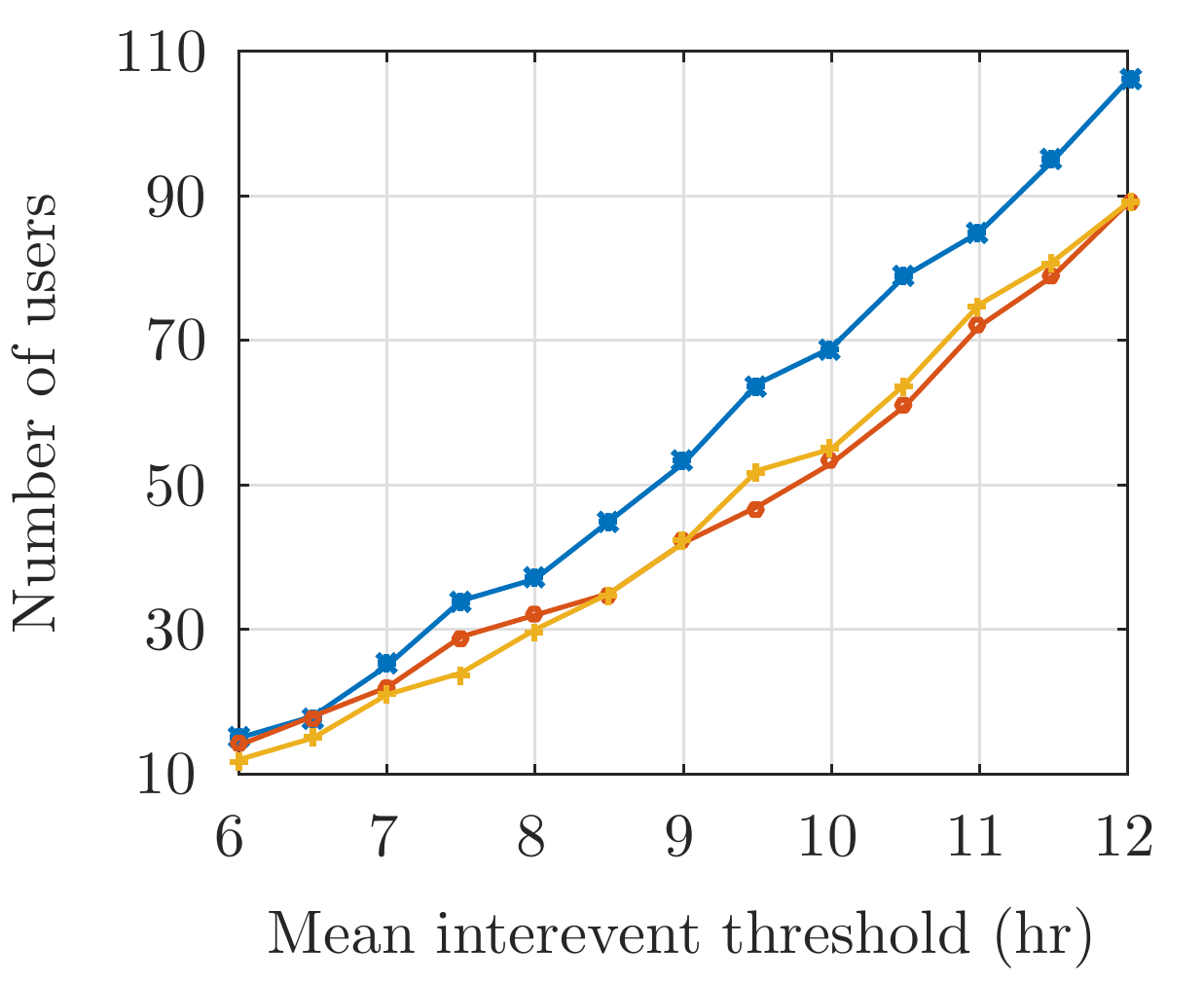}
	\includegraphics[width=0.325\textwidth]{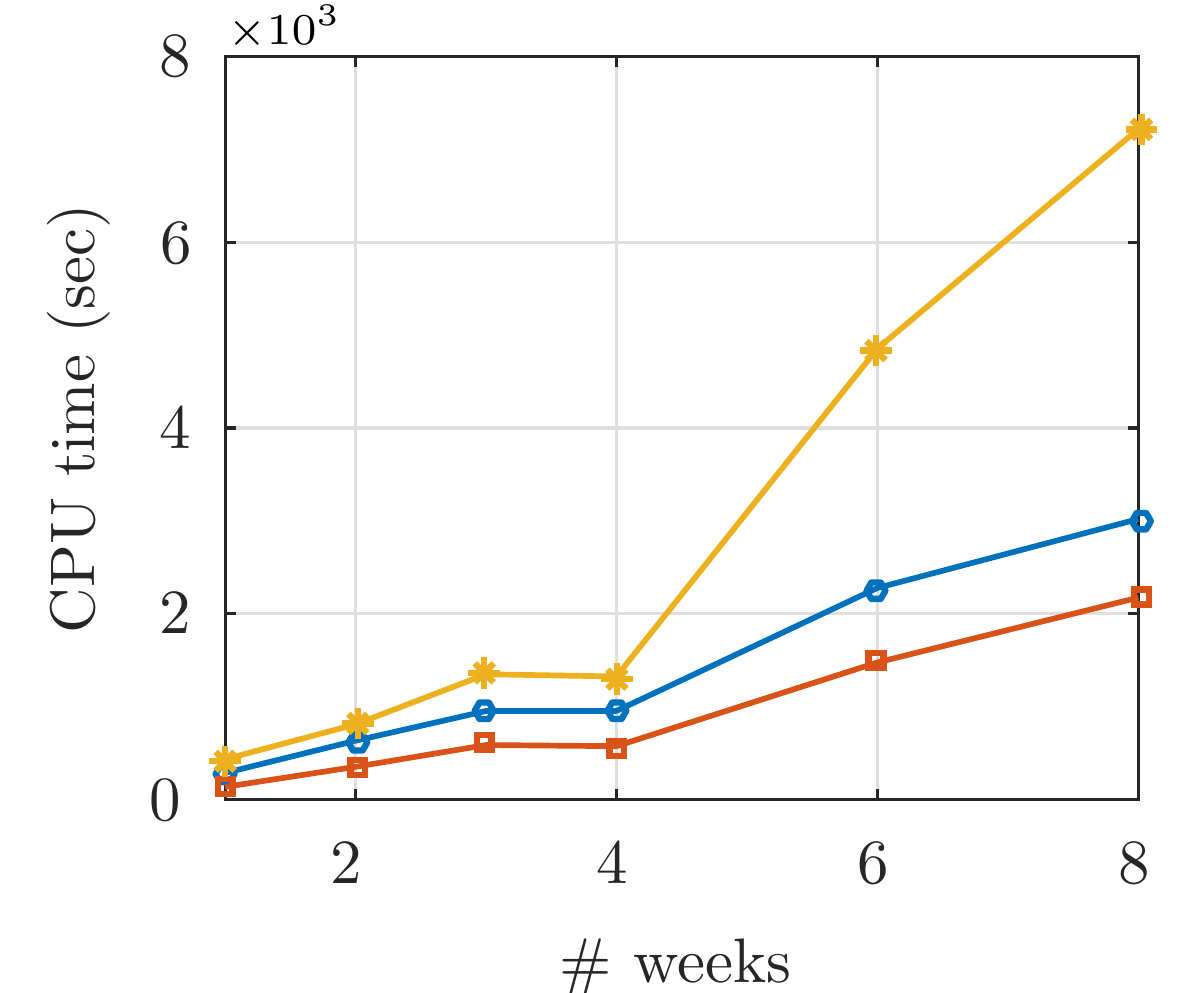}
	\caption{The percent of check-ins which their times are closer than a threshold to the real time (\textit{left}). The number of users which their average distance of predicted check-in times to the real times are less than a threshold (\textit{middle}). The time complexity of different temporal models (\textit{right}).}
	\label{fig:real-temporal}
\end{figure}
\begin{figure}[t]
	\centering
	\includegraphics[height=4.5mm]{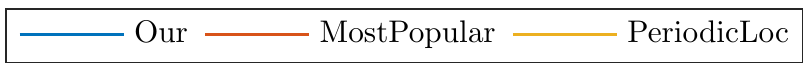}\\
	\vspace{-0.5mm}		
	\includegraphics[width=0.35\textwidth]{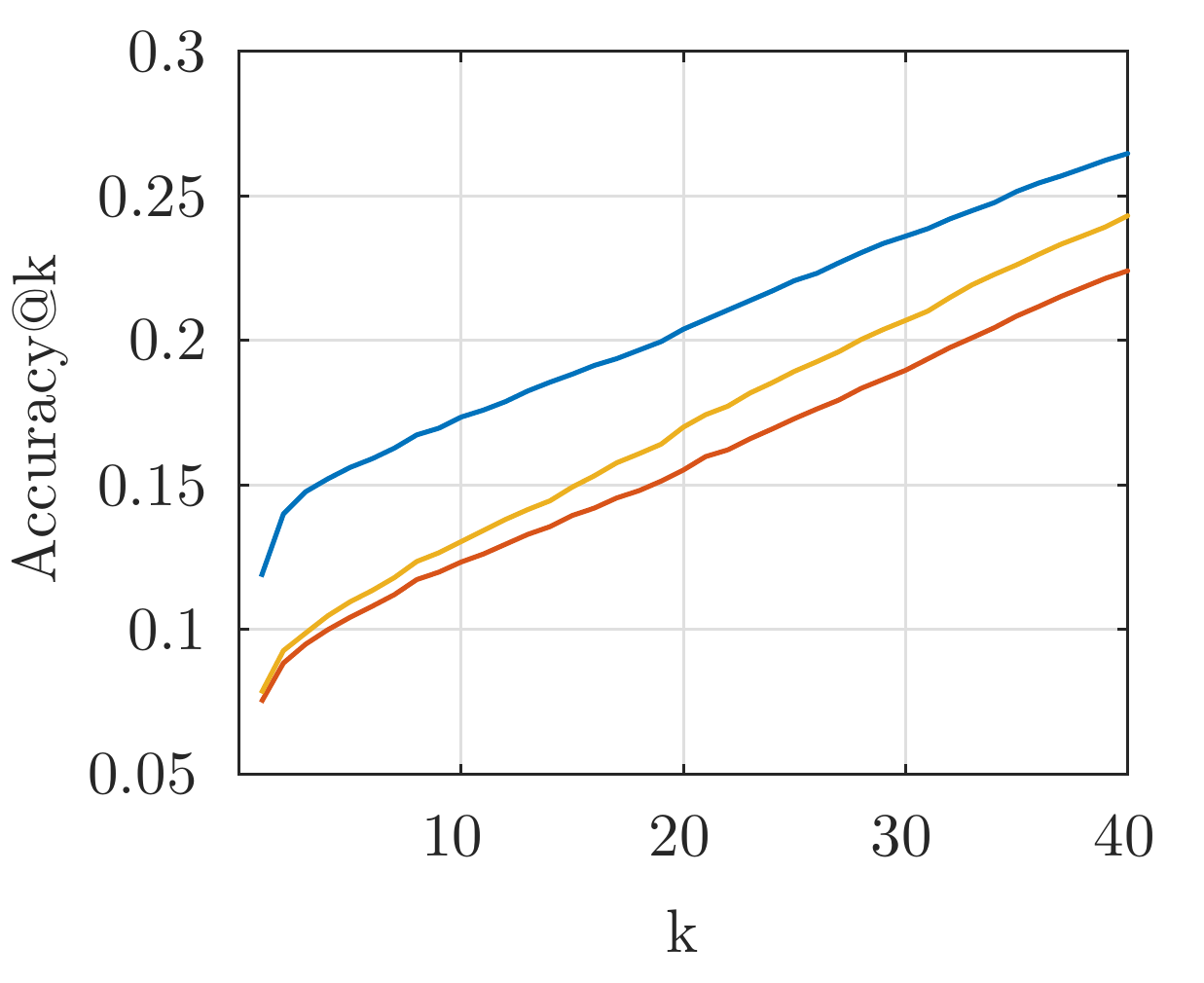}
	\hspace{0.5cm}
	\includegraphics[width=0.35\textwidth]{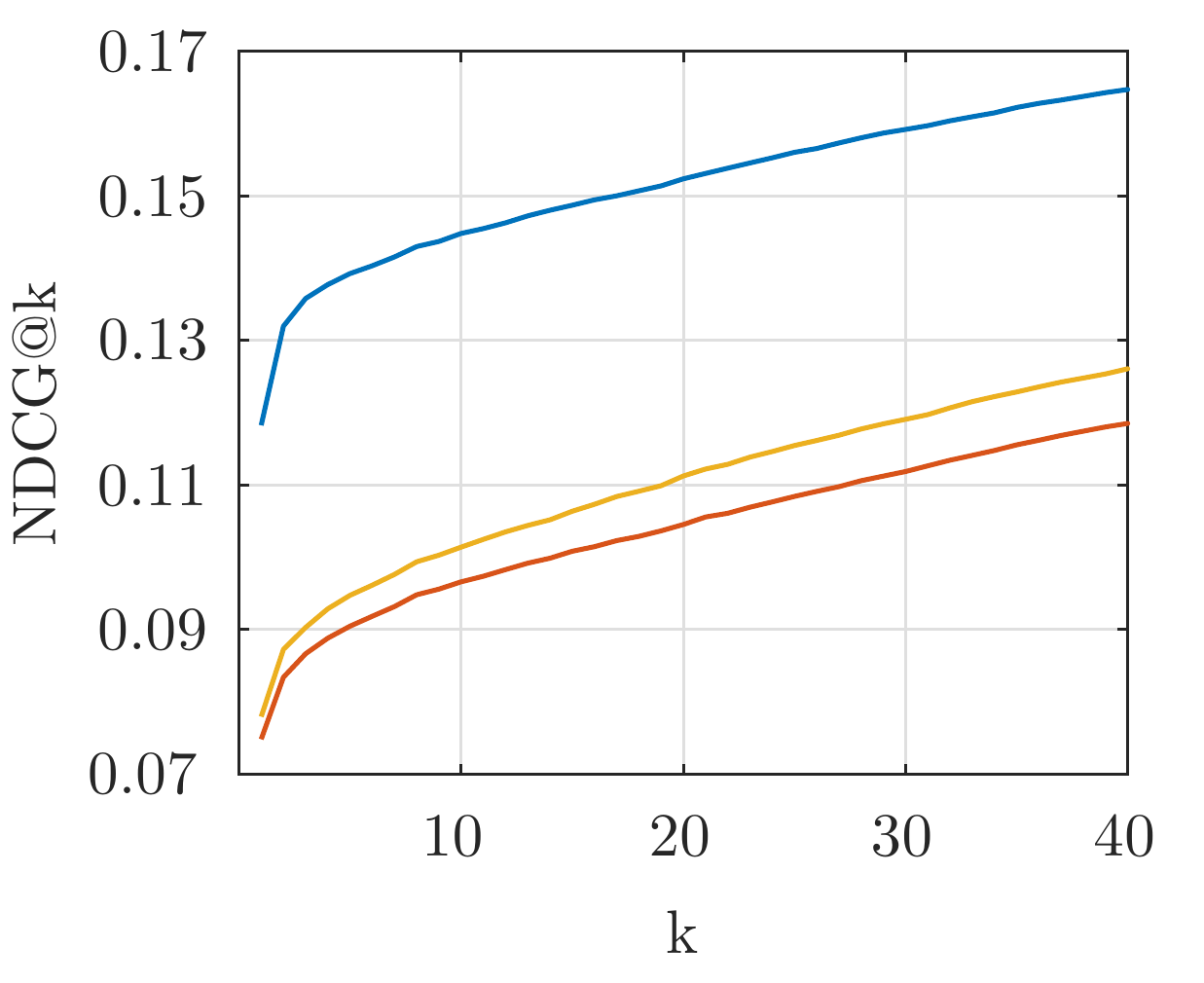}
	\caption{The accuracy (\textit{left}) and NDCG (\textit{right}) of location prediction, given the times of check-ins, at different values of $k$.}
	\label{fig:real-spatial}
\end{figure}

\subsubsection{Results}
To reveal the motivation of the proposed method, we perform two empirical experiments on the real data. In summary, Fig. \ref{fig:real-empirical} shows that: (i) most of the events are repeated after one, or more days (since there are peaks in the left graph at  $1,2,3,\ldots$), which verifies the use of a periodic point process for modeling the time of users' check-ins; (ii) about $80\%$ of users are affected by their friend's location of check-ins (blue box plot) which justifies the use of the proposed mutually-exciting spatial model; (iii) only $10\%$ of users explore new locations (red box plot), which these users are modeled by the parameter $\eta$ in Eq. (\ref{equ:spt-spatio}); (iv) as we more increase  the size of the history time window, the less \textit{Sociality} increases, which validates the use of the exponential decaying kernel in Eq. (\ref{equ:spt-weight-w}) to reduce the effect of far past history. 

To evaluate the prediction accuracy of check-in times we design two experiments. 
For each test event we estimate the time of the next event by different methods. The percent of check-ins which their times are closer than a threshold to the real time is plotted in the left graph of Fig. \ref{fig:real-temporal}. Our method achieved up to $35\%$ improvement for a $1$ hr threshold  compared to other methods. In the middle graph, the number of users where the average distance of their estimated events is less than a threshold is plotted. The proposed method performed up to $20\%$ better than the competing methods.
We did not plot this graph for the thresholds less than $6$ hr, where all methods perform poorly. Finally, the right graph of Fig. \ref{fig:real-temporal} shows that the time complexity of our method is near the fastest method.

Now, given the time of check-ins, we evaluate the prediction accuracy of the location of check-ins. For each test event, each method assigns a probability to each location, forming a vector and selects the most probable location. \textit{Accuracy}$@$\textit{k} is the percent of events that the true location is among the first $k$ high probable locations, and  \textit{NDCG}$@$\textit{k} is $\frac{1}{N}\sum_{i=1}^N {\mathbb{I}(1+r(e_i)<k)}/{\log_2(r(e_i))}$, where $r(e_i)$ is the (one-based) rank of the real location of $i$'th check-in in the location probability vector. These measures are plotted in Fig. \ref{fig:real-spatial}. For $k=1$ the accuracy increase from ${\scriptstyle\sim}7\%$ in other methods to ${\scriptstyle\sim}11\%$ in our method\footnote{It should be noted that there are about 10,000 locations and the random guess has extremely low accuracy.}---about $43\%$ improvement. 
For larger values of $k$ the measure is less reliable, since all method would have the same accuracy. 
Our method arrived at $24\%$ accuracy, and about $8\%$ improvement at $k=40$. But in the \textit{NDCG} which dose not have the mentioned undesirable effect (since the low-rank events are more significant) we see our method consistently outperform the others---about $30$-$50\%$ improvement for the different values of $k$.

\section{Conclusion}

To model the check-ins of users in location-based social networks, we proposed a periodic point process for the time of check-ins, which leverages the periodicity in users' behavior, and a multinomial distribution for the location of check-ins, which leverages the mutually-exciting effect of friends on the decision of users.

The synthetic experiments show the proposed inference algorithm can learn the model parameters with high accuracy and its performance increases by the size of train data. Moreover, we study the effect of model parameters on the users' check-ins, from which one can interpret the users' behavior in LSBNs from their inferred parameters.
The real experiments on the curated Foursquare check-ins dataset, show the proposed method outperform the other competing methods in the time and location prediction of users' future check-ins. 
Specifically, we achieved up to $35\%$ in the time prediction and $43\%$ in the location prediction accuracy. Furthermore, the empirical studies  show the real data meets the assumptions of the proposed model that is, users are periodic in the time and mutually-exciting in the location of their checkins.

Our work also opens many interesting venues for future works. For example, we can consider the home location of the users in defining the probability of the location of their check-ins, by modifying the weight of locations in Eq. (\ref{equ:spt-weight-w}). In addition, we can investigate the utilization of a non-parametric spatial model instead of the multinomial distribution. Finally, we can use the proposed model to control the check-in behavior of users by incentivization.

\bibliographystyle{ACM-Reference-Format}
\bibliography{ref}

\end{document}